\documentclass{article}

\usepackage{fullpage}
\usepackage{ifthen}
\usepackage{xspace}
\usepackage{amsmath,amssymb}
\usepackage{algorithm}
\usepackage[noend]{algorithmic}
\usepackage{enumitem,tikz,dsfont,comment}
\usepackage{subcaption}

\usepackage{hyperref}
\hypersetup{
	colorlinks=true,
	linkcolor=magenta,
	filecolor=magenta,
	citecolor = magenta,      
	urlcolor=magenta,
}

\usepackage[capitalize]{cleveref}

\usepackage[suppress]{color-edits}
\addauthor[David]{dk}{red}
\addauthor[Bobby]{rdk}{blue}
\addauthor[Sid]{sb}{purple}

\newtheorem{theorem}{Theorem}
\newtheorem{lemma}{Lemma}
\newtheorem{corollary}{Corollary}
\newtheorem{proposition}{Proposition}
\newtheorem{example}{Example}

\newcommand{\Omit}[1]{}

\newenvironment{rtheorem}[3][]{%
\noindent \ifthenelse{\equal{#1}{}}{\textsc{#2 #3}}{\textsc{#2 #3 (#1)}}
\begin{it}
}{\end{it}}

\providecommand{\Omit}[1]{}

\providecommand{\half}{\ensuremath{\frac{1}{2}}\xspace}
\providecommand{\third}{\ensuremath{\frac{1}{3}}\xspace}
\providecommand{\quarter}{\ensuremath{\frac{1}{4}}\xspace}

\providecommand{\dd}{\,\mathrm{d}}
\providecommand{\eps}{\ensuremath{\varepsilon}\xspace}

\DeclareSymbolFont{AMSb}{U}{msb}{m}{n}
\DeclareMathSymbol{\N}{\mathord}{AMSb}{"4E}
\DeclareMathSymbol{\Z}{\mathord}{AMSb}{"5A}
\DeclareMathSymbol{\R}{\mathord}{AMSb}{"52}

\providecommand{\SET}[1]{\ensuremath{\{ #1 \}}\xspace}
\providecommand{\Abs}[1]{\ensuremath{| #1 |}\xspace}
\providecommand{\Set}[2]{\ensuremath{\SET{#1 \mid #2}}\xspace}
\providecommand{\Ind}[1]{\ensuremath{\mathds{1}[#1]}\xspace}

\providecommand{\PROB}{\ensuremath{{\mathbb P}}\xspace}
\providecommand{\Prob}[2][]{\ensuremath{%
\ifthenelse{\equal{#1}{}}{\PROB[#2]}{\PROB_{#1}[#2]}}\xspace}

\providecommand{\Expect}[2][]{\ensuremath{%
\ifthenelse{\equal{#1}{}}{\mathbb{E}}{\mathbb{E}_{#1}}%
\left[#2\right]}\xspace}
\providecommand{\ExpectC}[3][]{\Expect[#1]{#2\;|\;#3}}
\providecommand{\Event}[2][]{\ensuremath{\ifthenelse{\equal{#1}{}}{\mathcal{#2}}{\mathcal{#2}_{{#1}}}}\xspace}

\newcommand{\PCDF}{\ensuremath{\Psi}\xspace}
\newcommand{\PCdf}[1]{\ensuremath{\PCDF(#1)}\xspace}

\newcommand{\PCdfI}[1]{\ensuremath{\PCDF^{-1}(#1)}\xspace}

\newcommand{\TCDF}[1][]{\ensuremath{%
\ifthenelse{\equal{#1}{}}{F}{F_{#1}}}\xspace}
\newcommand{\TCdf}[2][]{\ensuremath{\TCDF[#1](#2)}\xspace}
\newcommand{\TPDF}[1][]{\ensuremath{%
\ifthenelse{\equal{#1}{}}{f}{f_{#1}}}\xspace}
\newcommand{\TPdf}[2][]{\ensuremath{\TPDF[#1](#2)}\xspace}

\newcommand{\TCDFINT}[1][]{\ensuremath{%
\ifthenelse{\equal{#1}{}}{\mathcal{F}}{\mathcal{F}_{#1}}}\xspace}
\newcommand{\TCdfInt}[2][]{\ensuremath{\TCDFINT[#1](#2)}\xspace}

\newcommand{\OPT}{\ensuremath{H^*}\xspace}
\newcommand{\Opt}[1]{\ensuremath{\OPT(#1)}\xspace}

\newcommand{\GCDF}[1][]{\ensuremath{%
\ifthenelse{\equal{#1}{}}{G}{G_{#1}}}\xspace}
\newcommand{\GCdf}[2][]{\ensuremath{\GCDF[#1](#2)}\xspace}

\newcommand{\Rank}[2][]{\ensuremath{%
\ifthenelse{\equal{#1}{}}{R(#2)}{R_{#1}(#2)}}\xspace}
\newcommand{\RankP}[2][]{\ensuremath{%
\ifthenelse{\equal{#1}{}}{R'(#2)}{R'_{#1}(#2)}}\xspace}
\newcommand{\RankPP}[2][]{\ensuremath{%
\ifthenelse{\equal{#1}{}}{R''(#2)}{R''_{#1}(#2)}}\xspace}
\newcommand{\WinPass}[2][]{\ensuremath{%
\ifthenelse{\equal{#1}{}}{w^+_{#2}}{w^+_{#1,#2}}}\xspace}
\newcommand{\WinFail}[2][]{\ensuremath{%
\ifthenelse{\equal{#1}{}}{w^-_{#2}}{w^-_{#1,#2}}}\xspace}

\newcommand{\Inv}[1]{\ensuremath{I(#1)}\xspace}

\newcommand{\FP}[1][]{\ensuremath{%
\ifthenelse{\equal{#1}{}}{\phi}{\phi_{#1}}}\xspace}
\newcommand{\PMA}{\ensuremath{\delta_a}\xspace}
\newcommand{\PMB}{\ensuremath{\delta_b}\xspace}
\newcommand{\PMBY}{\ensuremath{\delta_{b,Y}}\xspace}
\newcommand{\PMBX}{\ensuremath{\delta_{b,X}}\xspace}
\newcommand{\PMT}{\ensuremath{\delta_{\theta}}\xspace}

\newcommand{\PM}[1][]{\ensuremath{%
\ifthenelse{\equal{#1}{}}{\delta}{\delta_{#1}}}\xspace}
\newcommand{\CP}{\ensuremath{\gamma}\xspace}

\newcommand{\LEFTF}[1][]{\ensuremath{%
\ifthenelse{\equal{#1}{}}{\ell^{-}}{\ell^{-}_{#1}}}\xspace}
\newcommand{\RIGHTF}[1][]{\ensuremath{%
\ifthenelse{\equal{#1}{}}{\ell^{+}}{\ell^{+}_{#1}}}\xspace}

\newcommand{\Unif}[2]{\ensuremath{\text{Uniform}[#1,#2]}\xspace}

\newcommand{\setR}{\ensuremath{\mathbb{R}}\xspace}

\newenvironment{proof}[1][\textsc{Proof.}]{#1 }{\hfill $\Box$} 

\usepackage{natbib}
\bibpunct[, ]{(}{)}{,}{a}{}{,}%
\def\bibfont{\small}%

\begin{document}

\title{Threshold Tests as Quality Signals:\\ Optimal Strategies, Equilibria, and Price of Anarchy}

\author{Siddhartha Banerjee\thanks{Cornell University, sbanerjee@cornell.edu}
\and David Kempe\thanks{University of Southern California, David.M.Kempe@Gmail.com}
\and Robert Kleinberg\thanks{Cornell University, rdk@cs.cornell.edu}}

\maketitle

\begin{abstract}
We study a signaling game between two firms competing to have their product chosen by a principal.
The products have (real-valued) qualities, which are drawn i.i.d.~from a common prior.
The principal aims to choose the better of the two products, but the quality of a product can only be estimated via a coarse-grained \emph{threshold test}:
given a threshold $\theta$, the principal learns whether a product's quality exceeds $\theta$ or fails to do so.

We study this selection problem under two types of interactions.
In the first, the principal does the testing herself, and can choose tests optimally from a class of allowable tests.
We show that the optimum strategy for the principal is to administer \emph{different} tests to the two products: one which is passed with probability \third and the other with probability $\frac{2}{3}$.
If, however, the principal is required to choose the tests in a
symmetric manner (i.e., via an i.i.d.~distribution),
then the optimal strategy is to choose tests whose probability of passing is drawn uniformly from $[\quarter, \frac{3}{4}]$.

In our second interaction model, test difficulties are selected endogenously by the two firms.
This corresponds to a setting in which the firms must commit to their testing (quality control) procedures before knowing the quality of their products. 
This interaction model naturally gives rise to a \emph{signaling game} with two senders and one receiver.
We characterize the unique Bayes-Nash Equilibrium of this game, which happens to be symmetric.
We then calculate its Price of Anarchy in terms of the principal's probability of choosing the worse product.
Finally, we show that by restricting both firms' set of available thresholds to choose from, the principal can lower the Price of Anarchy of the resulting equilibrium;
however, there is a limit, in that for every (common) restricted set of tests, the equilibrium failure probability is strictly larger than under the optimal i.i.d.~distribution.

\end{abstract}

\section{Introduction} \label{sec:introduction}
A principal wants to choose between two firms producing interchangeable products, whose qualities are drawn i.i.d.~from a known prior.
The principal wants to pick the product of higher quality --- however, she cannot directly observe the products' qualities.
In order to learn more about the products' qualities, the principal can simultaneously subject the products to \emph{tests}.
Specifically, we consider the simplest and most coarse-grained tests: binary (i.e., pass/fail) threshold tests that reveal whether the product's quality lies above or below a chosen $\theta$.
How should the principal choose the tests to administer to the two products, so as to help her maximize the probability of picking the better of the two?
We refer to this as the \emph{optimal selection} problem.

Now consider an alternative setting in which firms conduct their own quality control in-house, according to a fully disclosed and verifiable procedure.
This may be necessary if the principal does not possess the expertise to conduct quality control herself.
In this setting, while the principal may not be able to conduct a test, we assume that she can verify that a firm correctly followed its disclosed testing protocol; in other words, we assume that firms inherently have the power to \emph{commit} to a test.
At the time a firm commits to a testing protocol, it will not know the exact quality of each individual product --- for example, due to variations across batches and over time, or because the firm acts as an intermediary (e.g. head hunters who vet candidates for a hiring firm).
Indeed, such variation is the reason testing is needed in the first place.
As before, we assume that firms have independent common priors for their product qualities.
How will firms choose tests in such an \emph{endogenous selection} setting, if each firm wants to maximize the probability of its own product being selected?
Will competition push the firms to subject themselves to very difficult tests, or will they coordinate on easy tests at equilibrium?
How much worse off is the principal due to having to outsource quality control tests, rather than conducting them herself?
Can she improve her probability of choosing the better product by restricting the set of tests from which the firms can choose, e.g., by prescribing standards that such tests must adhere to?


Endogenous test selection by two firms can be naturally viewed as a form of signaling; committing to a testing procedure takes the role of committing to a signaling scheme.%
\footnote{This is the more common view of signaling in the economics community: a signaling scheme is interpreted as a device (physical or otherwise) that maps relevant states of the world to observable signals. Fixing a device constitutes committing to a signaling scheme. In contrast, recent works in computer science apply signaling/persuasion to scenarios such as communications where it is less clear whether the sender has the ability to commit to a mapping.}
Thus, our work can be construed as a natural \emph{game} played between two agents whose strategies are signaling schemes from a restricted class of available schemes.
This parallels several recent works on Bayesian persuasion games between multiple firms vying for customers \citep{au:kawai:competitive-disclosure,au:kawai:multiple-senders,boleslavsky:carlin:cotton:exaggeration,boleslavsky:cotton:limited-capacity,hwang:kim:boleslavksy:competitive-advertising};
we discuss these in detail in~\cref{sec:related-work}.
Our high-level question is what the equilibria of such signaling games look like, and how much efficiency is lost (if any) by letting the agents/firms choose their own signaling scheme rather than the principal being able to control how she receives information about the state of the world.

We investigate such questions using the following simple model (see~\cref{sec:preliminaries}).
The two firms have products with real-valued \emph{qualities} $X,Y$ drawn randomly from a common prior with continuous cdf \PCDF.
The principal has at her disposal a collection of tests parametrized by a threshold $\theta \in \setR$ which encodes the difficulty level of each test.
When a firm's product with quality $X$ is subjected to a test with threshold $\theta$, the outcome reliably reveals whether $X \geq \theta$ (the product passes the test) or $X < \theta$ (the product fails the test).
In the language of signaling, this means that we restrict to signaling schemes with binary outcomes, in which the sets mapped to each outcome are intervals.

Based on the chosen test difficulties (which are observable in both optimal and endogenous selection regimes) and their outcomes, the principal selects one of the products.
Her objective is to minimize the probability of choosing the worse product, while each firm's objective is to maximize the probability of having its product chosen. We consider the following models, which endow the principal with varying degrees of control:
\begin{enumerate}[nosep]
\item The principal must give both firms the same test. \label{enum:same-test}
\item The principal has full control over the difficulties $\theta_X, \theta_Y$ of the tests given to the two firms. \label{enum:full-control}
\item The principal specifies a distribution from which both firms draw tests in an i.i.d.~manner. The restriction to \emph{identical} distributions may be required to achieve ex-post fairness, compared to, for instance, randomizing which of the two firms gets which of two non-identical tests.\label{enum:iid-distribution}
\item The firms may endogenously choose their own tests via equilibrium strategies. \label{enum:equilibrium}
\item The principal can restrict available tests to a set $S$ (common to both firms), and firms endogenously choose their tests from $S$. Such a restriction could arise if the principal is a government agency or sufficiently powerful firm providing binding quality control guidelines.\label{enum:restricted-equilibrium}
\end{enumerate}

It is clear --- simply from suitable subset relationships on sets of available actions --- that in terms of the principal's error probability,
$\left\{\text{\ref{enum:same-test}},\text{\ref{enum:equilibrium}}\right\} \geq \text{\ref{enum:restricted-equilibrium}} \geq \text{\ref{enum:iid-distribution}} \geq \text{\ref{enum:full-control}}$.
Our goal is to explicitly characterize the optimal or equilibrium outcomes under these five models, thereby inferring which of the preceding comparisons are strict, as well as to quantify the increase in error probability for the principal resulting from a move to a weaker model. When comparing a model in which the principal has control with one in which the agents are allowed to choose tests according to an equilibrium strategy, this ratio exactly corresponds to the Price of Anarchy.

\subsection{Other Applications and Model Discussion}
While we phrase our work in terms of two firms offering products, our model applies more broadly.
In particular, it can be viewed as a generalization of the classic ``forum shopping'' model of \citet{lerner:tirole:forum-shopping} to multiple firms (\emph{property owners}, in their language).%
\footnote{\citet{lerner:tirole:forum-shopping} do briefly discuss a multi-firm setting, but only consider one extremely limited example.}
In this model, firms can choose an external certification agency to issue a recommendation on whether or not their product is ``acceptable.''
There is a continuum of agencies, ranging from fully aligned with the firm's interests to fully independent. Under suitable parameters, this model precisely corresponds to being able to choose any quantile threshold for a test.
While the model does not place the tests ``in house,'' in terms of the firms' choices, it is equivalent to our model.
The focus of \citet{lerner:tirole:forum-shopping} is on the interplay of the independence/difficulty of the agency and the owner's ``concessions'' --- direct transfers to any user of the property, such as price reductions or additional features.
As they argue, such a setup not only captures agencies certifying products, but also journals/conferences reviewing papers and similar endeavors.
In addition to these applications, some of the literature on multi-sender cheap talk/Bayesian persuasion is motivated in terms of competing proposals, either to a funding agency or internal within an organization; see, e.g., \citep{boleslavsky:cotton:limited-capacity,boleslavsky:carlin:cotton:exaggeration}.

Another application, aligned with the classic work of \citet{spence:job-signaling} and \citet{ostrovsky:schwarz:disclosure}, is in the assessment of students. Here, the test is a pass/fail exam (or class) via which a student is assessed. The optimization problem may guide a teacher aiming to correctly rank the students in a class, while the endogenous test selection model roughly corresponds to students choosing the difficulty of projects to undertake or of classes to enroll in.

In the context of applications, three key assumptions in our model are worth discussing.
The first is that firms are unaware of their quality when choosing tests.
This power of commitment \emph{before} the state of the world is revealed is the defining distinction between Bayesian Persuasion and Cheap Talk models, and is covered in depth in~\cref{sec:related-work}. As we discuss, most works on inter-firm signaling make this assumption. For example, \citet{lerner:tirole:forum-shopping} assume that property owners do not know users' utilities for their product.\footnote{However, we note that in addressing the same real-world scenario, \citet{gill:sgroi:optimal-choice} instead consider a model where the owner knows the state before choosing the certifier; see~\cref{sec:related-work} for details.}
Similarly, \citet{ostrovsky:schwarz:disclosure} consider early contracting between students and employers, in which students at the time of negotiation only have priors on their future performance.
Naturally, as with all models, this assumption is a simplification, with reality lying between full and no commitment power. 

The second assumption is that tests have binary and monotone (i.e., pass/fail) outcomes; in particular, we assume that no test can be passed with quality $x$, but failed with quality $x' > x$.
Restricting to monotone information structures is quite common in the literature: for recent examples, see \citep{dworczak:martini:simple-economics}, \citep{onuchik:ray:conveying-value}, and \citep{candogan:spillovers} and discussions therein.
Other kinds of restricted signal spaces also have significant precedent in the literature. \citet{LimitedSignaling} analyze Bayesian Persuasion in which the sender is restricted in terms of the number of signals. \citet{boleslavsky:carlin:cotton:exaggeration} assume that the state of the world is binary (the product is good or bad) and allow each sender to only send one of two signals; nevertheless, competition between senders results in complex signal distributions at equilibrium.
Similarly, the certification models of \citep{gill:sgroi:optimal-choice,lerner:tirole:forum-shopping} mostly consider binary outcomes (recommend/don't recommend). As argued in \citep{gill:sgroi:optimal-choice} (see, e.g., Footnote~3 in \citep{gill:sgroi:optimal-choice} and the literature cited there), the main purpose of a test or evaluation is to provide a \emph{concise} summary of the product. 
When the outcome of the evaluation must be concise, the number of possible signals that can be sent is necessarily bounded, and a binary signal is a clean and idealized way to capture such a desideratum. Monotonicity is natural to assume when signals should be interpretable by a decision maker.
This justification is also borne out by the coarse-grained grading systems (pass/fail, grades A--F) typically used in education contexts.
It also closely aligns with the argument made in \citet{sobel:giving-receiving-advice}
that there is a tradeoff between accuracy and complexity of advice (i.e., signals). 

The third assumption is that there are exactly two firms (for most of our results), and that their qualities are drawn i.i.d. This assumption is very standard in the study of related questions in competitive signaling; see, e.g., the in-depth discussion of
\citep{li:rantakari:yang:competitive-cheap-talk,boleslavsky:cotton:limited-capacity,boleslavsky:carlin:cotton:exaggeration,hwang:kim:boleslavksy:competitive-advertising}
and additional related work in~\cref{sec:related-work}.
We discuss the difficulties with extending the result to $n > 2$ firms or non-identical priors in~\cref{sec:conclusions}.

\subsection{Our Results}
As we elaborate in~\cref{sec:preliminaries}, it is equivalent --- and much more convenient --- to characterize tests not in terms of their thresholds, but in terms of the probability that a product will fail the test. Thus, we can view each possible test as a real number in $[0,1]$; in this case, the products' qualities can be assumed w.l.o.g.~to be drawn \emph{uniformly} from $[0,1]$.

When both firms' products have to be subjected to the same test, it is easy to see that the optimum test is the \emph{median} test, passed with probability exactly \half, which chooses the wrong product with probability \quarter (see~\cref{sec:preliminaries}). When the principal can give the firms different tests, our main result is summarized by the following theorem. (See~\cref{sec:optimal-correlated,sec:optimal-iid} for formal statements.)

\begin{theorem}[Optimal Selection of Tests by Principal: Informal]
\hspace{1cm}
\begin{enumerate}[nosep]
\item If the principal can assign arbitrary tests to the two firms, then it is optimal to give one firm a test of \third and the other a test of $\frac{2}{3}$. This results in a probability of $\frac{1}{6}$ of incorrect selection.
\item If the principal must draw \emph{i.i.d.}~tests for the firms, then the optimal rule draws test thresholds uniformly from the interval $[\quarter, \frac{3}{4}]$. This results in a probability of $\frac{5}{24}$ of incorrect selection.
\end{enumerate}
\end{theorem}

The preceding theorem is rather surprising! Even though the firms' products have i.i.d.~qualities, the principal can decrease her failure probability significantly (by 33\%) by giving the firms very different tests.
Analogously, a teacher trying to optimally rank students by ability should give the students different tests, even if their abilities share a common prior distribution.

\smallskip

For the case of endogenous test selection, the equilibrium and its probability of a mistake are characterized by the following result, stated formally and proved in~\cref{sec:equilibrium}:

\begin{theorem}[Equilibrium Distribution]
When firms' qualities are drawn i.i.d.~uniformly from $[0,1]$, and firms choose their test difficulties endogenously, there is a unique Bayes-Nash Equilibrium, which is symmetric, and consists of each firm choosing difficulty $\theta \in [0,1]$ from the probability density function (pdf) $\TPdf{\theta} = \frac{1}{2 (\theta^2 + (1-\theta)^2)^{3/2}}$.

  The principal's resulting probability of incorrect selection is approximately 0.23056, causing a Price of Anarchy of approximately 1.38336 compared to the optimum correlated tests and approximately 1.10653 compared to the optimum i.i.d.~test distribution.
\end{theorem}

Finally, in~\cref{sec:equilibrium-restricted}, we allow the principal to set ``guidelines'' for the firms' quality control tests, by prescribing a set $S \subseteq [0,1]$ from which the thresholds must be drawn.

\begin{theorem}[Restricted Equilibrium Distribution]
  When the firms' qualities are drawn i.i.d.~uniformly from $[0,1]$, and the firms choose their test difficulties endogenously from an interval $S = [a,b] \subseteq [0,1]$, there is a unique Bayes-Nash Equilibrium. This unique Bayes-Nash equilibrium is symmetric and can be explicitly characterized in closed form.

  Moreover, there exist values $a,b$ for which the resulting probability of a mistake by the principal is strictly smaller than for the interval $[0,1]$; for example, for the interval $[0,0.79]$, the probability of a mistake is approximately 0.22975.

  However, even compared to a principal restricted to i.i.d.~test choices, under symmetric Bayes-Nash Equilibria, the Price of Anarchy is lower-bounded by a constant strictly larger than one: for every set $S \subseteq [0,1]$ (not just intervals), the probability of a mistake is at least $\frac{5}{24} + \frac{1}{82944}$.
\end{theorem}

One interesting interpretation of the preceding theorem is that a somewhat bigger part of the problem with endogenous test selection is that firms skew too much towards harder tests. Making extremely difficult tests (the top 20\%) unavailable results in a (slightly) better equilibrium probability for the principal. However, as we will see in the analysis, when restricting the interval of available tests, the equilibrium distribution has non-trivial point mass at the upper end of the interval; in other words, at equilbrium, firms will still compete by choosing difficult tests.

A visual representation of our results is given in~\cref{fig:failure-probabilities}.
Taken together, our theorems imply a strict separation of all five models of test selection, and notably show that the principal has a higher probability of incorrect selection when choosing the same test for both firms compared to when they choose tests endogenously.

\begin{figure}[!th]
\begin{center}
\scalebox{0.75}{
\begin{tikzpicture}[auto,thick,-]

  \draw[|-|] (15,0) -- (28,0);
  \draw (15,0.3) node {0.15};
  \draw (28,0.3) node {0.28};
  
  \draw[magenta] (16.666,0.5) to (16.666,0);
  \draw[magenta] (16.666,0.8) node {correlated};

  \draw[cyan] (20.8333,-0.5) to (20.8333,0);
  \draw[cyan] (20.8333,-0.8) node {i.i.d.};
  
  \draw[red,thin,|-|] (20.8345,0.3) to (22.975,0.3);
  \draw[red] (21.9,0.6) node {restricted eq.};
  
  \draw[blue] (23.056,0) to (23.056,-0.5);
  \draw[blue] (23.056,-0.8) node {unrestricted eq.};
  
  \draw[brown] (25,0) to (25,0.5);
  \draw[brown] (25,0.8) node {identical};

\end{tikzpicture}
}
\end{center}
  
\caption{\em The principal's failure probabilities under different models of threshold choices. \label{fig:failure-probabilities}}
\end{figure}
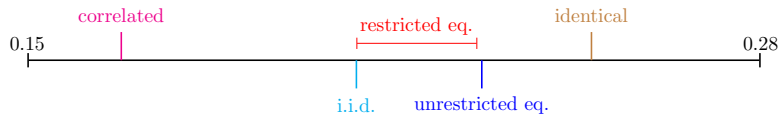

\dkcomment{Brought this paragraph back --- I assume it was deleted for space reasons. I also brought a few other comments/footnotes back for the same reason.}

Our work raises a wealth of directions for future inquiry, discussed in detail in~\cref{sec:conclusions}. Most immediate would be extensions to more than two firms and to richer signaling schemes. For an extension to multiple firms, an important point is to decide what the principal's and the firms' objectives are. One natural generalization is to have the principal still choose one (or $k$) of the firms' products; this appears difficult. A ``friendlier'' generalization involves a principal who wants to fully rank the firms by quality (e.g., a teacher in a classroom setting), and aims to minimize the number of inversions compared to the true order. In this setting, a firm/student may try to minimize the expected number of other firms ranked ahead of it. Because the objective functions naturally decompose into pairwise objectives by linearity of expectation, our results carry over to this setting completely. The only necessary generalization is for the case of correlated tests. In fact, in~\cref{sec:optimal-correlated}, we characterize the optimal choice of tests for the principal in the presence of any number of firms.

\section{Related Work in Depth}
\dkcomment{Moved this back to front.}
\label{sec:related-work}
\subsection{Multi-Sender Signaling}

Our equilibrium analysis can be considered as a case of multi-sender
signaling to a single receiver, in which the senders are constrained
to using threshold strategies to report on their (real-valued) type.
There is a significant body of work on signaling models with multiple
senders.
Two primary dimensions along which to categorize this body of work is
the extent to which the senders have commitment power,
and how much of the state of the world each sender can observe.
When the senders have commitment power, and do not know the state at
the time they commits to a strategy, we obtain the Bayesian persuasion
model of \citet{kamenica:gentzkow:bayesian}.
(See also algorithmic results by \citet{dughmi:information-structure,dughmi:xu:algorithmic,hssaine2018information}.)
On the other hand, when the senders have no commitment power, the
interaction is more accurately modeled as ``Cheap Talk'' in the sense
of \citet{crawford:sobel:strategic}.
The second dimension has two natural extremes: senders observe the
entire state of the world, or only information relevant to their own
``product.''
The former assumption is more relevant in the context of policy advice,
while the latter applies to competition between firms.
See the survey by \citet{sobel:giving-receiving-advice},
though most of the survey is focused on a single sender,
and the survey predates much of the recent literature on Bayesian
Persuasion.

\subsubsection{Cheap Talk Models}

Much of the early multi-sender literature focused on cheap talk
involving multiple senders with access to the full state of the world.
The primary motivation was the study of lobbying, advocacy, and
outsourcing of the acquisition of expert advice.
\citet{krishna:morgan:model-expertise}
and \citet{battaglini:multiple-referrals} showed that contrary to the
classical Cheap Talk model of \citet{crawford:sobel:strategic},
when there are multiple competing senders,
full revelation of the state of the world can be achieved even when
the interests of the receiver and the senders are not aligned.
\citet{matthews:postlewaite:quality-testing} and
\citet{milgrom:roberts:relying} studied a model in which each sender
can introduce uncertainty only by specifying a \emph{set} containing
the actual state of the world. A skeptical receiver can always force
full revelation, and \citet{milgrom:roberts:relying} show that in some
cases, competition between senders also leads to full revelation,
even when the receiver is not skeptical.

Incentives for information acquisition by multiple senders are studied
by \citet{austen-smith:wright-competitive-lobbying}
and \citet{dewatripont:tirole:advocates}.
Their primary focus is on when/whether it is beneficial for a
principal to outsource the acquisition of information to senders who
may either be exogenously biased towards one outcome,
or who can be induced by the principal, via a suitable rewards
structure, to advocate in favor of one outcome over the other.
A different angle of incentives for information acquisition is studied
by \citet{brocas:carrillo:palfrey:gatekeepers}
and \citet{gul:pesendorfer:war-of-information}.
In their model, each of two senders wants to convince the receiver
that the state of the world matches a particular value.
Each sender can pay to obtain another i.i.d.~signal biased towards the
\emph{true} state of the world.
The focus is on the implications of the tradeoff between costs for
signal acquisition and welfare achieved under suitable strategies.

To the best of our knowledge, the only work on competition in a cheap
talk setting is by \citet{li:rantakari:yang:competitive-cheap-talk}.
In their model, the two senders each have i.i.d.~uniform quality in
$[0,1]$, and can send messages to the receiver, who chooses exactly
one of them.
Each sender's objective is to maximize the quality of the selected
sender, plus an (additive or multiplicative) bonus for having oneself
selected.
The main result is that, similar to the equilibrium for
the single-sender Cheap Talk model \citep{crawford:sobel:strategic},
each sender still partitions $[0,1]$ into intervals and simply reveals
the interval containing the quality.
The number of intervals of the senders differs by at most 1,
and each interval intersects at most two of the other sender's
intervals.
\citet{li:rantakari:yang:competitive-cheap-talk} use this
characterization to analyze which of the multiple equilibria are best
for the receiver.

\subsubsection{Bayesian Persuasion Models}

Within the framework of Bayesian Persuasion,
\citet{gentzkow:kamenica:competition,gentzkow:kamenica:rich-signal-spaces}
study a model in which the information set is \emph{Blackwell-connected},
meaning that each sender can unilaterally deviate to \emph{every}
information state in which the receiver has more information.
For practical purposes, this means that each sender has access to the
full state of the world. The focus in
\citep{gentzkow:kamenica:competition,gentzkow:kamenica:rich-signal-spaces}
is on showing that competition instead of collusion between the
senders leads to more information being revealed to the receiver,
and on analyzing the effects of additional senders and more competition.
\citet{li:norman:sequential-persuasion} study similar questions in a
setting where the senders commit to their strategies sequentially
rather than simultaneously.

Most directly related to our work is a body of literature studying
direct competition between multiple senders in a Bayesian Persuasion
framework. The typical setup is that each of $n$ (in most work, $n=2$)
senders has a product or proposal whose quality is drawn from a
commonly known distribution over a set (often the set
$\SET{\text{good}, \text{bad}}$);
the realization of the quality draw is private to the sender.
The principal (receiver) can choose (at most or exactly) one
of the products. Each sender can choose an arbitrary information
disclosure policy (i.e., Bayes-plausible signaling scheme) about his
quality, and the goal is to analyze the equilibrium outcomes of this
game.

\citet{boleslavsky:cotton:limited-capacity} study a model in which
each project is either good or bad, independently and with known prior
probabilities. The receiver has a ``reserve expected quality'' and
will not accept projects whose expected quality lies below this
reserve. Even though the state space is binary and the receiver has
at most 3 actions (which would imply a small finite number of signals
in a single-sender setting), the competitive nature results in senders
choosing from among a larger number of signals.
\citet{boleslavsky:cotton:limited-capacity} analyze the equilibria,
and show that competition leads to more information disclosure
compared to a case when the receiver can choose both projects.
The more accurate decisions resulting from the additional information
in many cases outweigh the receiver's loss of utility from being able
to accept only one proposal.

\citet{au:kawai:competitive-disclosure} extend the model of
\citet{boleslavsky:cotton:limited-capacity} to allow for positive
correlations between the qualities of the projects. They identify two
ways in which positive correlation can affect the senders' utilities
(and hence strategies), and show that when the two senders' prior
quality estimates are very different, in the limit, large correlation
leads to more information disclosure.

\citet{au:kawai:multiple-senders} extends the analysis of
\citet{boleslavsky:cotton:limited-capacity} to $n > 2$ senders,
and also allows for larger (though still finite) sets of qualities.
\citet{au:kawai:multiple-senders} characterize the equilibrium
distribution in terms of payoff distributions.
One of the key results is that as the number of senders grows, the
equilibrium in the limit is full disclosure by all senders.

\citet{boleslavsky:carlin:cotton:exaggeration} study a model in which
the agents/senders can choose investments in project success probability.
The cost to the agent increases quadratically in the desired project
success probability.
The principal can observe the agents' investments, but not whether the
projects are actually good.
For the latter, the agents can commit to signaling schemes.
As in our work, \citet{boleslavsky:carlin:cotton:exaggeration}
restrict the signaling schemes to map to binary outcomes.
In other words, agents can commit to distributions with which they
will exaggerate their projects' successes, but cannot send more
differentiated fractional messages.
Among the key observations in \citep{boleslavsky:carlin:cotton:exaggeration}
is the fact that the ability to exaggerate (as opposed to being forced
to reveal whether projects were successful) typically leads to higher
investments by the agents.
Notice that similar to our work, the model of
\citet{boleslavsky:carlin:cotton:exaggeration} also involves mapping a
continuous variable to a binary signal.
Different from our model, however, the agents explicitly control the
investment.
Furthermore, the quantity that matters to the principal is the coin
flip itself (i.e., whether the project was successful), whereas in our
case, the continuous quality itself is what matters.

\citet{hwang:kim:boleslavksy:competitive-advertising} study a model in
which the receiver is a customer who will buy the product of exactly
one of the firms/senders.
The customer's values for the products are drawn i.i.d.~from a
commonly known distribution. In addition to the information disclosure
policy (i.e., signaling scheme), firms also control their prices.
The equilibrium characterization shows a strong connection to the
convexity/concavity of the value distribution. It alternates intervals
of full disclosure with those of uniform randomization, depending on
whether the value function is concave or convex in that interval.
As part of their analysis,
\citet{hwang:kim:boleslavksy:competitive-advertising} also show that
for a fixed price with convex value distribution, full disclosure
results.

\subsection{Other Signaling Work}

\dkcomment{Moved these two paragraphs down, since they are about single-sender settings.}
Two related papers on Bayesian persuasion with a single sender are \citep{LimitedSignaling,letreust:tomala:persuasion-communication};
both are among the few papers which explicitly impose communication constraints on the sender.
\citet{LimitedSignaling} restrict the number of signals that the sender may use, and study both welfare and algorithmic implications of such restrictions.
Much of the focus is on a model of price discrimination by a seller
who is informed by the sender about the buyer's type.
\citet{LimitedSignaling} also show that in general,
the best signaling strategy with limited signals is NP-hard to
approximate to within any constant.
\dkedit{\citet{letreust:tomala:persuasion-communication} study a Bayesian persuasion game that is repeated many times, where communication takes place over a limited and noisy information channel. The focus of the analysis is on the loss of sender utility arising because of the limited communication.}

The notion of receivers taking a test to determine their unknown quality is also considered in a very different context by~\citep{hssaine2018information}.
They study the selective disclosure of such scores in a multi-\emph{receiver} Bayesian signaling setting, where the principal aims to influence the formation of teams among the receivers.
The motivation in that work arises from semi-collaborative competitions such as formation of homework groups or teams for online coding challenges and crowdsourcing competitions.
The main assumption is that when participating in any such competition, an agent may not fully know his skill level; however, the principal can determine the skill level via an entrance test whose scores are visible only to her.
The principal can then exploit this information asymmetry to manipulate the agents' posteriors in order to make them form more diverse teams.

Several papers discuss somewhat less standard models of signaling product quality.
\citet{hoffmann:inderst:ottaviani:selective-disclosure}
consider a model in which the utility of each sender's product to the
receiver is the sum of two i.i.d.~terms,
and each sender must disclose exactly one of these terms.
The sender can choose whether to disclose a random term or the higher
of the two, and the analysis distinguishes whether the
receiver is aware of the senders' strategies.
\citet{hoffmann:inderst:ottaviani:selective-disclosure} consider this
as a simple model of information collection for targeted advertising,
and study whether the required data collection and targeted
advertising (with or without the consumer's knowledge) is in the
consumer's interest.
When senders cannot commit to their strategy, the unique equilibrium
(which also maximizes the consumer's utility) involves each sender
revealing the higher of the two terms.
When the senders \emph{can} commit to revealing a random term,
then revelation of the higher term by all senders becomes the unique
equilibrium as the number of senders grows large, but may not be an
equilibrium for few senders.

\citet[Chapter 3]{libgober:thesis} studies a model in which each of $n$
agents has a set of candidate projects, whose utilities in case of
success or failure are commonly known. Agents privately observe the
success probabilities of their projects, and communicate information
to the principal, who chooses an agent and a project to pursue.
The focus in \citep{libgober:thesis} is on simpler strategy
spaces, in which each agent selects only one of the projects to
propose to the principal, rather than revealing information about all
probabilities.
\rdkcomment{I tried to find Libgober's working paper
on his website but it's not there.}
\dkcomment{I'm citing his thesis instead now.}

A different model of delegated project selection
was analyzed by \citet{kleinberg_delegated_2018}, who considered
a single agent sampling $n$ candidate projects and choosing to
propose one to a principal, who may or may not choose to accept
the proposal. The agent's and principal's utilities for a project
may differ, and the focus in \citep{kleinberg_delegated_2018} is
on designing mechanisms whose equilibrium outcome for the principal
is approximately as good as if the agent's and principal's utilities
were identical. Unlike in the model we study here, the model in
\citep{kleinberg_delegated_2018} assumes that both the agent
and the principal can directly assess the utility of a proposed project;
the information asymmetry comes from the fact that the agent
evaluates $n$ candidate projects,
but the principal only evaluates the project the agent proposes.

In its focus on signals with binary outcomes,
our work also closely relates to the literature on external
certification, in particular the work of
\citet{lerner:tirole:forum-shopping}.
Their model, while phrased differently, is mathematically equivalent
to the following: a property owner has a property\footnote{such as a
  product, research proposal, or scientific paper} of unknown value
drawn from a commonly known distribution.
The owner can choose a certifier\footnote{such as a standards agency
  or journal} who will verify whether the value
lies above or below a threshold --- the assumption is that for each
$\alpha$, there is a certifier with threshold exactly $\alpha$.
Conditioned on the outcome, users will buy the property if the
conditional expected value lies above a known reserve.
An important additional feature in \citep{lerner:tirole:forum-shopping}
is the ability of the owner to directly transfer utility to users
buying the property in the form of ``concessions''\footnote{such as
  price reductions, additional features, or added figures or results},
and the primary focus of \citet{lerner:tirole:forum-shopping} is a
study of the interplay between the selection of certifier and the
resulting concessions.
Without concessions, the model of \citet{lerner:tirole:forum-shopping}
directly corresponds to our model of endogenous test selection for a
single firm.
\citet{lerner:tirole:forum-shopping} also briefly discuss an extension
to multiple property owners, but only consider one very specific
example.
Our work can be considered a more general treatment of the 2-owner
setting without a reserve utility, in which users simply want to
select the better of the two properties.

\citet{gill:sgroi:sequential-decisions,gill:sgroi:optimal-choice}
also study a model of certification and its impact on product
acceptance.
In their model, the state of the world (product quality) is binary
(high or low).
Furthermore, the owner is aware of the state of the world before
choosing a certifier, and needs the external certifier because he
cannot credibly transmit the state himself.
Certifiers have different (and known) accuracies and difficulties, and
the focus is on how various model parameters affect the choice of
certifier and related actions.
In \citep{gill:sgroi:sequential-decisions}, the owner aims to maximize
the probability of a herding cascade on his product (when agents
observe each others' decisions).
\citep{gill:sgroi:optimal-choice} instead focuses on the interplay
between the choice of a certifier and adjustments to the price after
the certifier's assessment is publicly revealed.
However, neither of
\citep{gill:sgroi:sequential-decisions,gill:sgroi:optimal-choice}
study a multi-owner scenario and the resulting competition.

Finally, \citet{ostrovsky:schwarz:disclosure} study a multi-sender
multi-receiver signaling game between schools and employers.
Schools have students whose abilities are drawn independently from
known distributions, and can signal to employers (via the students'
transcripts) how good a student is.
Each school's objective is to maximize the expected desirability of
the students' employment, while each employer aims to maximize the
ability of the student hired.
\citet{ostrovsky:schwarz:disclosure} show that typically,
at equilibrium, schools will withhold some information.
Much of the subsequent focus is on the investigation of early
contracting, involving students signing employment contracts before
their full transcript is known;
\citet{ostrovsky:schwarz:disclosure} show that early contracting
will not occur if schools disclose the equilibrium amount of
information about their students, but can occur otherwise.

\subsection{Algorithmic Considerations for Ranking from Limited Data}

Apart from the related literature in economics and game theory, our study of optimal selection of tests also relates to a large algorithmic literature on ranking and selection based on limited information.
This is a core topic in data mining and Bayesian optimization, with a vast body of work; we briefly outline some ideas in this space which are closely related to our algorithmic approaches.

The most directly related topic in the data mining literature is that of \emph{learning-to-rank}, which refers to a general framework of constructing probabilistic ranking models for a set of objects, training these based on data, and then using them to sort new objects according to their degrees of relevance, preference, or importance. The classical model here is the ELO ranking; a more popular modern system is the TrueSkill system of~\citet{herbrich2006trueskill}, which is used by Microsoft and others for ranking online gamers. \citet{liu2009learning} provide an overview of this and related approaches to learning-to-rank, and their applications in information retrieval.

Learning-to-rank systems can further be subdivided into pointwise, pairwise and global ranking systems; in this context, our approach bears similarities with information-theoretic variants of ranking based on \emph{pairwise comparison} tests~\citep{jamieson2011active,negahban2017rank}. The main idea in these works is to consider obtaining noisy signals of pairwise comparisons between sets of items with a true underlying ordering, where the noise in each signal depends on the distance between the two underlying items.
Another related field which looks at similar models and questions regarding ranking under probabilistic models is that of social choice. Here again, our work has close connections to studies of statistical~\citep{shah2017simple,conitzer2006improved} and computational~\citep{betzler2009fixed,kenyon2007rank} properties of different ranking algorithms used in social choice theory.
The main difference in our treatment, however, is that we are able to design the tests we use for ranking, and also consider strategic aspects in agents choosing which tests to use (which then motivates considering rank aggregation through the lens of a signaling game).

With regards to the idea of endogenous selection games in ranking, our work shares commonalities with work of~\citet{altman2007incentive,altman2010axiomatic} on incentives in ranking systems. The main focus of these works is to obtain an axiomatic characterization of ranking systems under which agents are incentivized to reveal their true skill levels. The crucial difference here is that agents are aware of their own skills, in contrast to our setting, where agents must take a test to discover their true skill.

Finally, the rise of collaborative platforms and MOOCs has led to a recent upsurge of interest in the use of testing for selecting teams. In this context,~\citet{kleinberg2018team} look at the question of how test scores of multiple agents can be used to form teams whose output depends on a complex functiof of agents' joint utility profiles. On the other hand, \citet{johari2018exploration} consider in some sense a dual question wherein a principal observes the scores of different teams, with each score being a complex function of the utility profile of the agents in a team, and must use this to try and rank the agents. The focus in these works has primarily been on the computational challenges of learning rankings and/or forming teams based on such scores, in contrast to our focus on the strategic aspects in such settings.

\section{Model and Preliminaries} \label{sec:preliminaries}
\subsection{Qualities, Tests and Selection}
\label{ssec:skills-tests-rankings}

We consider a setting in which a principal wants to pick the better of the products provided by two firms $X$ and $Y$. We will equivalently refer to this process as \emph{selecting} or \emph{choosing} a firm or \emph{ranking} the firms.
The two firms' products have i.i.d.~\emph{qualities} $X, Y$ drawn from a common prior distribution with continuous cdf%
\footnote{We adopt the convention that the cumulative distribution function
  (cdf) of a probability measure on $\setR$ is defined by setting
  \TCdf{x} to be the measure of the set $(-\infty,x]$ under the distribution.}
\PCDF on \setR. 
Abusing notation, we use $X,Y$ to refer both to the firms themselves and their products' (random) qualities.

Information about the products' qualities is revealed by means of \emph{binary threshold tests} (henceforth simply \emph{tests}) administered to the products.
More specifically, a test is completely characterized by a threshold $\theta \in \setR$.
A product of quality $X$ subjected to a test with threshold $\theta$ \emph{passes} if and only if $X \geq \theta$; otherwise, we say that the product \emph{fails} the test $\theta$.
To avoid unnecessary clutter in writing, we also refer to the \emph{firm} $X$ or $Y$ as passing or failing the test (instead of its product).
The larger $\theta$, the less likely a product is to pass the test, so we can naturally think of $\theta$ as the \emph{difficulty} of the test.
When a product is subjected to a test, the outcome (pass or fail) is revealed to everyone, but no additional information can be inferred about the product.
This model is mathematically equivalent to the certification model of \citet{lerner:tirole:forum-shopping}.

The principal's goal is to minimize the probability of selecting the product of lower quality. We refer to this as an \emph{incorrect selection}, or as an \emph{error} by the principal, or --- by analogy with ranking --- as an \emph{inversion}.
Formally, consider a rule $\mathcal{T}$ for assigning tests to firms and selecting a firm based on the tests' outcome.
We define $\Inv{\mathcal{T}} := \Ind{\mathcal{T} \text{ chooses the wrong firm}}$ as the indicator of $\mathcal{T}$ inverting the ranking.
Note that \Inv{\mathcal{T}} is a random variable, with randomness arising from: (1) $\mathcal{T}$'s selection of test thresholds, (2) the firms' products' random qualities, and (3) possibly randomized aggregation of test outcomes.
The principal's goal is to choose $\mathcal{T}$ so as to minimize
$\Expect{\Inv{\mathcal{T}}}$.

Given a firm's test result, the principal can form a posterior belief of its product's quality.
The posterior expected quality of a product passing threshold test $\theta$ is $\ExpectC[X \sim \PCDF]{X}{X \geq \theta}$, while the posterior expected quality of a product failing it is $\ExpectC[X \sim \PCDF]{X}{X < \theta}$.
Observe that for any product quality cdf \PCDF, we have that $\ExpectC[X \sim \PCDF]{X}{X < \theta}$ and $\ExpectC[X \sim \PCDF]{X}{X \geq \theta}$ are monotone non-decreasing in $\theta$,
and strictly increasing for $\theta$ in the support of \PCDF.
Furthermore
\begin{align*}
  \ExpectC[X \sim \PCDF]{X}{X < \theta}
  & \leq \Expect[X \sim \PCDF]{X}
  \; \leq \; \ExpectC[X \sim \PCDF]{X}{X \geq \theta},
\end{align*}
and both inequalities are strict if $\theta$ is in the support of \PCDF.
Because both products' qualities are drawn from the same distribution, these observations imply the following proposition.


\begin{proposition} \label{prop:ranking-pass-fail}
  Let $\theta_X > \theta_Y$ be the thresholds of the tests to be applied to the products of firms $X,Y$. Assume that both $\theta_X, \theta_Y$ lie in the support of \PCDF.
  \begin{enumerate}[nosep]
     \item If both firms' products pass their tests, or both fail their tests, then the principal minimizes the probability of an inversion by selecting $X$.
     \item If exactly one of the products of $X, Y$ passes its test, then the principal minimizes the probability of an inversion by selecting the firm that passed.
  \end{enumerate}
\end{proposition}

\Cref{prop:ranking-pass-fail} characterizes a rational principal's choice (once test outcomes have been revealed) almost completely.
To complete the description, we assume that when there is a tie, the principal picks one of the firms uniformly at random. We will refer to this case as a coin flip, and say that $X$ (or $Y$) \emph{wins/loses the coin flip}.
As an illustration, consider the following example:

\begin{example}[The Median Test]
{\em  Suppose that both firms' products have i.i.d.~quality levels $X, Y\sim \Unif{0}{1}$ (i.e., drawn uniformly over $[0,1]$). A natural test is the \emph{median test} $\mathcal{T}_{\text{median}}$, under which both products are subjected to a test with $\theta=\half$.
  A product's posterior expected quality upon passing is $\ExpectC{X}{X \geq 1/2} = 3/4$, and upon failing $\ExpectC{X}{X \leq 1/2} = 1/4$.
  Now w.l.o.g.~suppose that the two firms' products have qualities $X < Y$. If $X \leq \half < Y$, then $Y$ passes and $X$ fails, and the principal ranks them correctly. However, if $Y \leq \half$, then both fail, and if $X > \half$, then both pass. In either case, a coin flip is required, and the principal chooses correctly only with probability \half.
  Thus, the median test achieves $\Expect{\Inv{\mathcal{T}_{\text{median}}}}=\quarter$.

  More generally, if the principal gives the same test $\theta$ to both agents, then an inversion happens if: (1) either both $X,Y \geq \theta$ or both $X,Y < \theta$, \emph{and} (2) the coin flip determines the wrong winner. Thus, the probability of an inversion is $\half (\theta^2 + (1-\theta)^2)$. This is minimized at $\theta=\half$, showing that the median test is optimal for the principal if she must give the same test to both agents.
}
\end{example}

Given complete control over the choice of testing rule $\mathcal{T}$, the principal's goal is to choose the rule that minimizes $\Expect{\Inv{\mathcal{T}}}$.
This could be a single threshold for both firms (as with the median test); a distribution \GCDF over $\setR$ such that, for each firm, the principal draws an i.i.d.~threshold $\theta \sim \GCDF$; or, most generally, a joint distribution \GCDF over the thresholds for both firms.
The optimal i.i.d.~threshold distribution and the optimal joint distribution are the subjects of~\cref{sec:optimal-iid,sec:optimal-correlated}.

\subsection{Endogenous Test Selection and Quantile Thresholds}
\label{sec:basic-properties}

In many settings, firms may be better equipped than the principal to perform quality control tests in house.
\footnote{%
Alternatively, the setting may be such that the agents naturally have the choice of test difficulty, such as in external certification of product quality \citep{lerner:tirole:forum-shopping,gill:sgroi:sequential-decisions,gill:sgroi:optimal-choice} or students' selection of which classes to attempt \citep{spence:job-signaling}.
In these settings, it is still frequently assumed that agents are \emph{not} aware of their private quality value when they make their choice of difficulty, see for example \citep{lerner:tirole:forum-shopping} for a model of certification and \citep{ostrovsky:schwarz:disclosure} for a model of contracting between students and employers.}
In these cases, the firms will typically commit to a verifiable quality control procedure for their products.
The principal gets to observe (only) the threshold $\theta$ and the outcome of the test.
In other words, both firms commit to a signaling scheme about their products' qualities, where the space of signaling schemes is restricted to a binary signal space and threshold functions.

Each firm's goal is to maximize its probability of being selected, or --- equivalently --- of being ranked ahead of the other firm.
Due to the competitive nature of the game, the appropriate solution concept (which we will study) is a \emph{Bayes-Nash Equilibrium}.
We refer to this setting as \emph{endogenous test selection}.
Because the firms are a priori symmetric, in any equilibrium, each firm's product must be selected with probability \half.

In a further generalization, note that the principal may be able to rule out some types of tests.
In other words, in a more general model, the principal may specify a closed set $S$ and restrict the firms to selecting test thresholds $\theta \in S$ only.
We will be primarily interested in the case when $S$ is an interval, but also consider more general closed sets $S$.

Before continuing, we note that since the utilities of both the principal and the firms depend only on rankings and not actual qualities, it is convenient to work in the quantile space $[0,1]$ rather than the quality space $\setR$.
To do this, note that for any quality $X \sim \PCDF$, its corresponding (random) \emph{quantile} \PCdf{X} is distributed uniformly in $[0,1]$.
Now, suppose that firm $X$ chooses (or is assigned) a threshold $\sigma \in \setR$ for its test;
we can equivalently view this as the firm picking a \emph{threshold quantile} $\theta = \PCdf{\sigma} \in [0,1]$.
Note that a product with quality $X \sim \PCDF$ passes a test with threshold quantile $\theta$ with probability $1-\theta$;
moreover, a threshold quantile $\theta \in [0,1]$ corresponds to a threshold $\sigma = \PCdfI{\theta}$ in the quality space,
where $\PCdfI{x} \triangleq \inf \Set{y\in\setR}{\PCdf{y}\geq x}$
is the generalized inverse function associated with the cdf \PCDF.
Thus, w.l.o.g., we henceforth focus on product qualities drawn from
$\PCDF \sim \Unif{0}{1}$, and understand ``threshold'' to refer
to the threshold quantile $\theta\in[0,1]$.

\subsection{Extension to More Firms}
\label{sec:multipleagents}

While we have focused thus far on the paradigmatic case of two firms,
the model can be naturally extended to $n \geq 2$ firms.
Several natural generalizations suggest themselves,
both in terms of the principal's objective and the firms' objective.
With $n$ firms, the principal may try to maximize the probability of choosing the best product, or try to produce a complete ranking of all firms' products, minimizing the total number of inversions.%
\footnote{There are naturally other objectives in between these two extremes.}
For a firm, the goal might be to maximize the probability of being selected, or to be ranked as highly as possible in expectation.
Our results extend naturally to the latter objectives, namely,
\begin{itemize}[nosep]
\item The utility of a firm is proportional to the number of firms ranked behind it.
\item The disutility of the principal is proportional to the (normalized) \emph{Kendall tau distance\footnote{\dkedit{Recall that the Kendall $\tau$ distance between two rankings is the number of inversions between the two rankings, i.e., the number of pairs of elements that are in different order.}}} between the true and inferred rankings, i.e., the fraction of pairwise inversions between the two lists.
\end{itemize}

Extending our notation from the case of two firms,
for a given rule $\mathcal{T}$ for choosing tests for firms, we denote the (random) Kendall tau distance between the resulting ranking and the correct ranking by \Inv{\mathcal{T}}.
Again, the principal's goal is to minimize $\Expect{\Inv{\mathcal{T}}}$.
Using linearity of expectations for both the firms and the principal, all of our results for two firms carry over immediately to the case of $n$ firms, with exactly the same guarantees regarding the fraction of misranked pairs.
The only exception is that for \emph{correlated tests} (in~\cref{sec:optimal-correlated}), the optimal choice for the principal will depend on the number $n$ of firms.
These results do not extend to other objectives, and both optimal and equilibrium strategies will typically look different for $n \geq 3$ firms.
See~\cref{sec:conclusions} for a discussion.


\section{Optimal I.I.D.~Tests} \label{sec:optimal-iid}
In this section, we explicitly characterize the optimal
distribution from which the principal should draw thresholds
when drawing them i.i.d.~for both firms.

\subsection{Characterizing the Expected-Inversions Functional}
\label{sec:objective-derivation}

Let $\mathcal{T}_{\GCDF}$ denote the test selection rule under which each firm is given a test with threshold drawn i.i.d.~from \GCDF.
We begin by characterizing the expected number of inversions as a functional of the cdf \GCDF from which the thresholds are drawn.
In the next section, we will show how to choose \GCDF to minimize this functional.
For notational convenience, we henceforth denote $\Inv{\GCDF} = \Expect{\Inv{\mathcal{T}_{\GCDF}}}$.

\begin{lemma} \label{lem:inversion-probability}
  \dkedit{Assume that the quality distribution \PCDF is uniform\footnote{\dkedit{Recall from \cref{sec:basic-properties} that this assumption is without loss of generality.}} on $[0,1]$.}
 Suppose that thresholds for both firms are drawn i.i.d.~from the distribution \GCDF on $[0,1]$ (not necessarily continuous).
  The probability of selecting the worse product is given by the functional
\begin{align}
\label{eqn:inversion-objective}
\Inv{\GCDF} & = \int_0^1 \int_0^x (1-\GCdf{x}+\GCdf{y})^2 \dd y \dd x.
\end{align}
\end{lemma}

\begin{proof}
Assume that the two firms' products have qualities $x > y$.
An inversion occurs when $Y$ is selected.
We can think of the process as first picking two thresholds
$\theta_0,\theta_1$ i.i.d.~from the distribution with cdf \GCDF,
letting $\theta = \min(\theta_0, \theta_1)$ and
$\theta' = \max(\theta_0, \theta_1)$,
and then uniformly randomizing which of $X$ and $Y$ gets which
of $\theta, \theta'$.
Because \GCDF may have point masses, it is possible that
$\theta = \theta'$.
We consider the following cases, based on which of the two or three
intervals $\mathcal{I}_1 = [0,\theta),
\mathcal{I}_2 = [\theta, \theta'),
\mathcal{I}_3 = [\theta', 1]$ the products' qualities $x$ and $y$ fall into
(see~\cref{fig:inversioncases}).

\begin{figure}[htb]
	\centering
	\resizebox{0.9\textwidth}{!}{
		\input{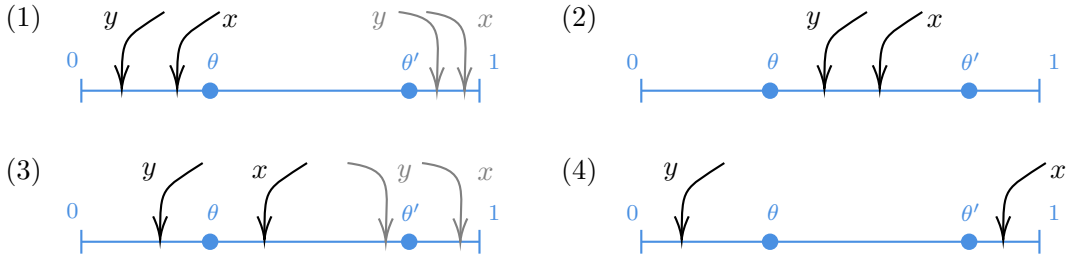}
	}
	\caption{\em The four cases for determining \Inv{\GCDF}. Cases (1) and (2) lead to inversions with probability \half, while in cases (3) and (4), the ranking of firms is correctly identified.}
	\label{fig:inversioncases}
\end{figure}
      
\begin{enumerate}[nosep]
\item If both $x, y \in \mathcal{I}_1$, then both firms fail, and if $x, y \in \mathcal{I}_3$,
  then both firms pass. In either case, an inversion happens with probability \half; if $\theta = \theta'$, this is due to the randomized tie-breaking rule, while for $\theta \neq \theta'$, it is because the assignment of $\theta, \theta'$ to the two firms is uniformly random, and the mechanism ranks as higher the firm which is assigned $\theta'$.
\item If both $x, y \in \mathcal{I}_2$, then the firm with threshold $\theta$ passes, while the one with $\theta'$ fails.
Because the assignment is uniformly random, again, an inversion is created with probability \half.
\item If $x \in \mathcal{I}_3$ and $y \in \mathcal{I}_2$, or
  $x \in \mathcal{I}_2$ and $y \in \mathcal{I}_1$, then no inversion
  can be created, because $X$ will always be ranked ahead of $Y$.
  This is because either $X$ passes and $Y$ fails, or otherwise, both
  firms obtain the same result and $X$ was assigned the higher
  threshold $\theta'$.
\item Finally, if $x \in \mathcal{I}_3$ and $y \in \mathcal{I}_1$,
  then $X$ passes and $Y$ fails, so no inversion is created.
\end{enumerate}

So an inversion is created with probability \half iff both $x$ and $y$
are in the same interval, and with probability 0 otherwise.
The probability that both are in $\mathcal{I}_1$ is
$(1-\GCdf{x})^2$ (because this case is equivalent to $\theta_0, \theta_1 > x$);
the probability that both are in $\mathcal{I}_2$ is
$2 \GCdf{y} \cdot (1-\GCdf{x})$ (because this case is equivalent to
$\theta_0 \leq y < x < \theta_1$ or $\theta_1 \leq y < x < \theta_0$);
and the probability that both are in $\mathcal{I}_3$ is
$\GCdf{y}^2$ (because this case is equivalent to
$\theta_0, \theta_1 \leq y$).
Thus, we have

\begin{align*}
  \ExpectC{\Inv{\mathcal{T}_{\GCDF}}}{(X,Y)=(x,y),x>y}
  & = \frac{\left( (1-\GCdf{x}^2) + \GCdf{y}^2 + 2\GCdf{y} (1-\GCdf{x}) \right)}{2}
= \frac{(1 - \GCdf{x} + \GCdf{y})^2}{2}.
\end{align*}

\noindent Therefore, the expected probability of an inversion overall is

\smallskip
$\Inv{\GCDF}
= \int_0^1 \int_0^1 \half (1-\GCdf{\max(x,y)} + \GCdf{\min(x,y)})^2
   \dd y \dd x 
\; = \; \int_0^1 \int_0^x (1-\GCdf{x}+\GCdf{y})^2 \dd y \dd x.
$
\end{proof}

\subsection{Optimizing the Objective Function}

We now provide a characterization of the i.i.d.~distribution \OPT that minimizes \Inv{\GCDF}.
\begin{theorem}
\label{thm:optimal-distribution}
\dkedit{Assume that the quality distribution \PCDF is uniform\footnote{\dkedit{Recall from \cref{sec:basic-properties} that this assumption is without loss of generality.}} on $[0,1]$.}
Let \OPT be the cdf corresponding to the uniform distribution over the interval $[\quarter, \frac{3}{4}]$.
For every distribution \GCDF over $[0,1]$,
we have $\Inv{\GCDF} \geq \Inv{\OPT}$.
\end{theorem}

In other words, the optimal way to pick i.i.d.~tests is to sample them \emph{uniformly from $[\tfrac14,\tfrac34]$}. This may seem somewhat surprising. Some intuition for this can be derived from looking at \emph{correlated} test selection rules in the limit of infinitely many firms (see the discussion after~\cref{thm:optimal-correlated} in~\cref{sec:optimal-correlated}).
The following proof provides a way to not only prove optimality of \OPT, but to obtain lower bounds on \Inv{\GCDF} for every \GCDF that differs substantially from \OPT.
We use this in~\cref{sec:equilibrium-restricted} to obtain lower bounds on the inversion probability of equilibria with restricted sets of available tests.

\begin{extraproof}{\cref{thm:optimal-distribution}}
We will show that the uniform distribution
on $[\frac14,\frac34]$ is the unique distribution
on $[0,1]$ that optimizes the functional \Inv{\GCDF}
defined by \eqref{eqn:inversion-objective}.
Let
\begin{equation} \label{eq:h0}
  \GCdf[0]{x} = \begin{cases}
    0 & \mbox{if } x < \tfrac14 \\
    2x - \tfrac12 & \mbox{if } x \in [\tfrac14,\tfrac34] \\
    1 & \mbox{if } x > \tfrac34
  \end{cases}
\end{equation}
be the cdf of the uniform
distribution on $[\frac14,\frac34]$, and let
\GCDF be (the cdf of) any other distribution on $[0,1]$.
For $t \in [0,1]$ we can consider the hybrid
distribution \GCDF[t] which draws a sample from \GCDF[0] with
probability $1-t$ and from \GCDF with probability $t$.
This hybrid distribution has cdf
\begin{equation} \label{eq:ht}
 \GCdf[t]{x} = t \GCdf{x} + (1-t) \GCdf[0]{x} = \GCdf[0]{x} + t (\GCdf{x} - \GCdf[0]{x}).
\end{equation}
We now prove that for every $\GCDF \neq \GCDF[0]$, the function \Inv{\GCDF[t]} is strictly increasing in $t$. This immediately implies that $\Inv{\GCDF} > \Inv{\GCDF[0]}$, confirming that $\GCDF[0] = \OPT$ is uniquely optimal, as claimed.

Substituting the right side of Equation~\eqref{eq:ht}
into the definition of $I(\GCDF[t])$, we have
\begin{align} \nonumber
  \Inv{\GCDF[t]} & =
    \int_0^1 \int_0^x \big( 1-\GCdf[0]{x}+\GCdf[0]{y}
       + t [\GCdf[0]{x} - \GCdf{x} + \GCdf{y} - \GCdf[0]{y}] \big)^2 \, \dd y \dd x \\
    & = \nonumber
    \Inv{\GCDF[0]} +
    2 t \int_0^1 \int_0^x (1 - \GCdf[0]{x} + \GCdf[0]{y}) \cdot
    (\GCdf[0]{x} - \GCdf{x} + \GCdf{y} - \GCdf[0]{y}) \, \dd y \dd x  \\
    & \quad +
    t^2 \int_0^1 \int_0^x \big( \GCdf[0]{x} - \GCdf{x} + \GCdf{y} - \GCdf[0]{y} \big)^2
    \, \dd y \dd x.\nonumber
\end{align}
The right side is a quadratic function
of $t$, i.e., we can write  $\Inv{\GCDF[t]} = \Inv{\GCDF[0]} +  A(\GCDF) \cdot t
+ B(\GCDF) \cdot t^2$, where
the coefficients of $t$ and $t^2$ are given by
\begin{align}
    A(\GCDF) &= 2 \int_0^1 \int_0^x (1 - \GCdf[0]{x} + \GCdf[0]{y}) \cdot
    (\GCdf[0]{x} - \GCdf{x} + \GCdf{y} - \GCdf[0]{y}) \, \dd y \dd x  \nonumber \\
  \label{eq:iht-b}
    B(\GCDF) &= \int_0^1 \int_0^x \big( \GCdf[0]{x} - \GCdf{x} + \GCdf{y} - \GCdf[0]{y} \big)^2
    \, \dd y \dd x .
\end{align}
If \GCDF and \GCDF[0] are not equal, then --- since they both are right-continuous --- they must differ on a set of positive measure.
Consequently, $B(\GCDF)$ is strictly positive.
To prove that \Inv{\GCDF[t]} is strictly increasing in $t$, we need only show that $A(\GCDF) \ge 0$.
We define
\begin{equation*}
  C(\GCDF) = 2 \int_0^1 \int_0^x ( 1 - \GCdf[0]{x} + \GCdf[0]{y} ) \cdot
    (\GCdf{x} - \GCdf{y}) \, \dd y \dd x,
\end{equation*}
and note that $A(\GCDF) = C(\GCDF[0]) - C(\GCDF).$
Since \GCDF is the
cdf of a random variable, it can be expressed as a convex combination of step functions. Specifically,
for $z \in [0,1]$, define the step function $T_z(x)=\Ind{x \ge z}$.
Now if $Z$ is a random sample from \GCDF,
then for all $x \in [0,1]$, we can write
$\GCdf{x} = \Expect[Z\sim\GCDF]{T_Z(x)}$.
Then observing that $C(\GCDF)$ is linear in \GCDF, we have via linearity of expectation that
$C(\GCDF) = \Expect[Z\sim\GCDF]{C(T_Z)}$.
Moreover, for any fixed $z \in [0,1]$, we have
\begin{align*}
  C(T_z) & = 2 \int_0^1 \int_0^x (1 - \GCdf[0]{x} + \GCdf[0]{y})
    (T_z(x) - T_z(y)) \, \dd y \dd x 
     =  2 \int_z^1 \int_0^z ( 1 - \GCdf[0]{x} + \GCdf[0]{y} ) \, \dd y \dd x,
\end{align*}
because for $y<x$, we have $T_z(x) - T_z(y) = 1$ when $y<z\leq x$, and $0$ otherwise. Also observe that
$1 - \GCdf[0]{x} = \GCdf[0]{1-x}$ (since \GCDF[0] is a distribution on $[0,1]$ and is symmetric about \half). Thus
\begin{align*}
  2 \int_z^1 \int_0^z 1 - \GCdf[0]{x} + \GCdf[0]{y} \, \dd y \dd x
  &= 2 \int_z^1 \int_0^z \GCdf[0]{1-x} + \GCdf[0]{y} \, \dd y \dd x \\
  = 2 \int_0^{1-z} \int_0^z \GCdf[0]{x} + \GCdf[0]{y} \, \dd y \dd x 
  &= 2 z \int_0^{1-z} \GCdf[0]{x} \, \dd x +
    2 (1-z) \int_0^z \GCdf[0]{y} \, \dd y \\
  &= 2 z \Upsilon(1-z) + 2 (1-z) \Upsilon(z),
\end{align*}
where the function $\Upsilon(\cdot)$ is defined by
\[
  \Upsilon(z) = \int_0^z \GCdf[0]{x} \, \dd x = \begin{cases}
    0 & \mbox{if } z < \tfrac14 \\
    \left( z - \tfrac14 \right)^2 & \mbox{if } z \in [\tfrac14,\tfrac34] \\
    z - \tfrac12 & \mbox{if } z > \tfrac34.
  \end{cases}
\]
For $z \in \left[ \frac14, \frac34 \right]$ we have
\begin{align*}
  z \Upsilon(1-z) + (1-z) \Upsilon(z) & =
  z \left( \tfrac34 - z \right)^2 + (1 - z) \left( z - \tfrac14 \right)^2
   = z \left( \tfrac{9}{16} - \tfrac32 z + z^2 \right) +
      (1-z) \left( \tfrac{1}{16} - \tfrac12 z + z^2 \right) \\
  & = \tfrac{1}{16} + \left( \tfrac{9}{16} - \tfrac{1}{16} - \tfrac{1}{2} \right) z
    + \left( -\tfrac32 + \tfrac12 + 1 \right) z^2 + (1 - 1) z^3
   = \tfrac{1}{16}.
\end{align*}
For $z < \frac14$ we have
$
z \Upsilon(1-z) + (1-z) \Upsilon(z) =
  z \left( \tfrac12 - z \right) =
  \tfrac{1}{16} - \left( \tfrac14 - z \right)^2,
$ 
while for $z > \frac34$ we have
$
z \Upsilon(1-z) + (1-z) \Upsilon(z) =
  (1-z) \left( z - \tfrac12 \right) =
    \tfrac{1}{16} - \left( z - \tfrac34 \right)^2.
$ 
Summarizing this case analysis,
$ 
  C(T_z) = 2\left( z \Upsilon(1-z) + (1-z) \Upsilon(z) \right) \leq \tfrac{1}{8},
$ 
and the inequality is strict if and only if
$z \not\in [\tfrac14,\tfrac34].$
Now, using the characterization that
$C(\GCDF) = \Expect[Z\sim\GCDF]{C(T_Z)}$ for every \GCDF, we have
$C(\GCDF) \leq \tfrac{1}{8}$,
and the inequality is strict if and only if
the support of \GCDF is not contained in
$[\tfrac14, \tfrac34]$.
On the other hand, since \GCDF[0] is the cdf of
a distribution supported on $[\tfrac14,\tfrac34]$, we have
$C(\GCDF[0]) = \tfrac{1}{8}$.
Combining these two inequalities,
we obtain
$A(\GCDF) = C(\GCDF[0]) - C(\GCDF) \geq 0$,
which demonstrates that \Inv{\GCDF[t]}
is strictly increasing in $t$, and thus
$
  \Inv{\GCDF} = \Inv{\GCDF[1]} > \Inv{\GCDF[0]}
$.
\end{extraproof}

We can now derive the optimal probability of inversion under
i.i.d.~tests.

\begin{corollary} \label{cor:optimal-cost}
  The optimum i.i.d.~test distribution creates an inversion with
  probability $\frac{5}{24}$.
\end{corollary}


\begin{proof}
Substituting \OPT into Eq.~\eqref{eqn:inversion-objective},
we get that the probability of an inversion under \OPT is
\begin{align*}
  \Inv{\OPT}
& = \int_0^1 \int_0^x (1-\Opt{x}+\Opt{y})^2 \dd y \dd x
\\ & = \int_0^{\quarter} x \dd x
     + \int_{\quarter}^{\frac{3}{4}} \quarter \left(1-2\left(x-\quarter\right)\right)^2 \dd x
     + \int_{\quarter}^{\frac{3}{4}} \int_{\quarter}^x (1-2(x-y))^2 \dd y \dd x
\\& \hspace{1.7cm}  + \int_{\quarter}^{\frac{3}{4}} \quarter \cdot \left(2\left(y-\quarter\right)\right)^2 \dd y
	 + \int_{\frac{3}{4}}^1 \left(x - \frac{3}{4}\right) \dd x
\\ & = \frac{1}{32} + \frac{1}{24}
     + \int_{\quarter}^{\frac{3}{4}} \int_0^{x-\quarter} (1-2y)^2 \dd y \dd x
     + \frac{1}{24} + \frac{1}{32}.
\end{align*}
Simplifying the expression further (with standard integration),
we get that $\Inv{\OPT}=\frac{5}{24}$.
\end{proof}

\section{Optimal Correlated Tests} \label{sec:optimal-correlated}
In \cref{sec:optimal-iid}, we derived the optimal distribution to sample tests from if each firm must be assigned a test \emph{independently} from the same distribution.
Here, we consider the problem when the firms' tests can be chosen \emph{in a correlated way}.

As we mention in \cref{sec:multipleagents}, although most of our analysis looks at two firms, it extends naturally to multiple firms when the goal is to minimize the expected number of inversions.
When the test assignments can be correlated, the actual number of firms affects the optimal solution.
Hence, in this section, we explicitly characterize the optimal choices when there are $n$ firms. Surprisingly, this takes the following simple form:

\begin{theorem} \label{thm:optimal-correlated}
  \dkedit{Assume that the quality distribution \PCDF is uniform\footnote{\dkedit{Recall from \cref{sec:basic-properties} that this assumption is without loss of generality.}} on $[0,1]$.}
Recall that \Inv{\mathcal{T}} denotes the (random) Kendall tau distance between the true and inferred rankings.
For $n$ firms, the expected fraction of inversions $\Expect{\Inv{\mathcal{T}}}$ is minimized over all correlated test selection rules $\mathcal{T}$ by one which assigns the test with threshold $\theta_i = \frac{n+2(i-1)}{4n-2}$ to firm $i$.
The resulting expected fraction of inverted pairs of firms is
  $\frac{5n-4}{12(2n-1)}$.
\end{theorem}

To get intuition for this result, it is instructive to consider it for $n=2$. In this case, the optimal $\mathcal{T}$ allocates two tests at thresholds $\frac{1}{3}$ and $\frac{2}{3}$, respectively, and this improves the fraction of misclassified pairs from $\frac{5}{24}$ to $\frac{1}{6}$.
The main reason behind this improvement is that the ability to give different
tests to the two firms allows the principal to choose tests to maximally split up the space $[0,1]$, such that the only way the principal makes a mistake is if the products' qualities $X,Y$ are in the same interval (refer again to \cref{fig:inversioncases}, cases (1) and (2)).

\Cref{thm:optimal-correlated} is also instructive in the limit as $n \to \infty$. Here, one sees that the optimal test distribution converges to uniformly spaced tests over the interval $[\quarter, \frac{3}{4}]$ (and leads to a $\frac{5}{24}$ fraction of pairs being inverted). This suggests that a uniform distribution of tests over $[\quarter, \frac{3}{4}]$ should be the optimal distribution for i.i.d.~tests for any number of firms, since drawing $n$ tests from a continuous distribution results in all $n$ tests being unique almost surely, and close to the optimal correlated tests. This intuition is indeed confirmed by the earlier \cref{thm:optimal-distribution}.

\begin{extraproof}{\cref{thm:optimal-correlated}}
For any fixed $n$-tuple of thresholds
$(\theta_1, \ldots, \theta_n) \in [0,1]^n$,
consider the expected (over the draw of the products' qualities)
number of inversions.
This quantity must have an actual minimizer,
which is the principal's optimal choice;
in other words, the thresholds need not be drawn from a distribution.

Because the products' qualities are drawn i.i.d., without loss of
generality, we assume $\theta_1 \leq \theta_2 \leq \cdots \leq \theta_n$.
Consider two firms $i < j$, and the partition of $[0,1]$ into the
three intervals $[0,\theta_i), [\theta_i, \theta_j), [\theta_j, 1]$.
If the qualities $x_i, x_j$ fall into distinct intervals, then
$i$ and $j$ will always be ranked in correct order,
as can be seen by the following cases
(similar to cases (3) and (4) in \cref{fig:inversioncases}):
\begin{enumerate}[nosep]
\item If $x_i < \theta_i$, then firm $i$ fails its test.
  Because firm $j$'s threshold is higher, whether it passes or fails,
  it will be ranked ahead of $i$,
  which is correct since $x_j \geq \theta_i > x_i$.
\item If $\theta_i < x_i < \theta_j$, then firm $i$ passes its test.
  If $x_j < \theta_i$, then $j$ fails its test and will be
  correctly ranked behind $i$.
  If $x_j \geq \theta_j$, it passes its test and will correctly be
  ranked ahead of $i$, because its test is more difficult.
\item If $x_i \geq \theta_j$, then firm $i$ passes its test.
  Because $x_j$ is in a different interval, $x_j < \theta_j$,
  so $j$ fails its test, and will be correctly ranked behind $i$.
\end{enumerate}

When both $x_i$ and $x_j$ are in the same interval, the outcomes
(fail/fail in the bottom interval, pass/fail in the middle, pass/pass
in the top) determine some ranking.
The actual ordering between $x_i$ and $x_j$, conditioned on being in
the same interval, is uniformly random, so the ranking is correct with
probability \half.
Thus, we have derived that the inversion probability for the pair $i,j$ is
$\half (\theta_i^2 + (\theta_j-\theta_i)^2 + (1-\theta_j)^2)$.
Write $(\theta_1, \ldots, \theta_n)$ for the rule that assigns each agent $i$ the threshold $\theta_i$.
Summing over all pairs $i < j$, the expected number of inverted pairs
for $\theta_1 \leq \theta_2 \leq \cdots \leq \theta_n$ is
\begin{align} \label{eq:opt-corr.1}
  \Expect{\Inv{(\theta_1, \ldots, \theta_n)}}
  & = \half \cdot \sum_{i=1}^n \sum_{j=i+1}^n \theta_i^2 + (\theta_j-\theta_i)^2 + (1-\theta_j)^2
  \\ & = (n-1) \sum_{i=1}^n \theta_i^2 + \half \cdot {n \choose 2}
       - \sum_{i=1}^n (i-1) \theta_i - \sum_{i=1}^n \sum_{j=i+1}^n
       \theta_i \theta_j.
\label{eq:opt-corr.2}
\end{align}
The right side of~\eqref{eq:opt-corr.1} is a
strongly convex quadratic function of $(\theta_1,\ldots,\theta_n)$,
and hence its global minimum over $\setR^n$ is attained
at the unique point where its gradient vanishes.
Using formula~\eqref{eq:opt-corr.2}
and setting the derivative with respect to all $\theta_i$ to be zero,
we have that
$ 
0 = 2(n-1) \theta_i - (i-1) - \sum_{j \neq i} \theta_j
= (2n-1) \theta_i - (i-1) - \sum_{j=1}^n \theta_j.
$ 
for all $i \in [n]$.
Writing $c = \sum_{j=1}^n \theta_j$, we get that
$\theta_i = \frac{c+i-1}{2n-1}$.
Therefore, $(2n-1) \cdot c = \sum_{j=1}^n (c+j-1)$,
which means that $c = \frac{1}{n-1} \cdot \sum_{j=0}^{n-1} j = \frac{n}{2}$, and thus $\theta_i = \frac{n+2(i-1)}{4n-2}$.
Thus, the vector $(\theta_1,\ldots,\theta_n)$
that minimizes the right side of~\eqref{eq:opt-corr.1}
over all of $\setR^n$ belongs to the
set $\{(\theta_1,\ldots,\theta_n) \mid 0 \leq \theta_1 \leq \theta_2 \leq \cdots \leq \theta_n \leq 1\}$, and therefore minimizes
$\Expect{\Inv{(\theta_1,\ldots,\theta_n)}}$
over that set.
Denote this test selection rule as $\mathcal{T}^*$.
Substituting the above choice of $\theta_i$ into
$\Expect{\Inv{(\theta_1, \ldots, \theta_n)}}$,
and omitting some simplifications,
we get that the expected number of inversions is

\begin{align*}
\Expect{\Inv{\mathcal{T}^*}} = & \frac{1}{4(2n-1)^2} \cdot \Big(
    (n-1) \sum_{i=1}^n (n+2(i-1))^2
    + 4(2n-1)^2 \cdot \frac{n(n-1)}{4}
\\ & \qquad \qquad \quad - 2(2n-1) \sum_{i=1}^n (i-1) (n+2(i-1))
    - \sum_{i=1}^n \sum_{j=i+1}^n (n+2(i-1)) (n+2(j-1)) \Big)
\\ = & \frac{n(n-1)(5n-4)}{24(2n-1)}.
\end{align*}
Dividing by ${n \choose 2}$,
we obtain that the \emph{fraction} of inverted pairs is $\frac{5n-4}{12(2n-1)}$.
\end{extraproof}

\section{Endogenous Test Selection and Price of Anarchy} \label{sec:equilibrium}
In this and the next section, we turn to the question of endogenous test selection. Here, we consider the setting where the principal makes all threshold tests in $[0,1]$ available to the firms for selection; in the next section, we consider the benefits of being able to restrict the set of offered tests.
\dkedit{For the entire section, recall that we assume without loss of generality (see \cref{sec:basic-properties}) that the quality distribution \PCDF is uniform on $[0,1]$.}

The particular equilibrium concept we study for endogenous test selection is that of a Bayes-Nash Equilibrium.
We say that a pair of distributions $(\TCDF[X],\TCDF[Y])$ supported on (a subset of) $[0,1]$ constitutes a Bayes-Nash Equilibrium of the \emph{endogenous test selection game} if, given that $X$ chooses a random test from $[0,1]$ according to \TCDF[X], choosing a test from \TCDF[Y] is a \emph{best response} for firm $Y$ (i.e., in the set of strategies that maximize $Y$'s selection probability), and similarly with the roles of $X$ and $Y$ reversed. The case when \TCDF[X] and \TCDF[Y] are identical is referred to as a symmetric Bayes-Nash Equilibrium. In this case, we will write $\TCDF = \TCDF[X] = \TCDF[Y]$ and refer to \TCDF as an equilibrium distribution, or simply an equilibrium.
We remind the reader that, as discussed in \cref{sec:basic-properties}, even though we focus on quality distributions being $\Unif{0}{1}$, the results extend naturally to any distribution \PCDF which is absolutely continuous.

The following proposition simply formalizes an observation from \cref{sec:preliminaries}: that at equilibrium, each of $X$ and $Y$ must be ranked first with probability \half.

\begin{proposition} \label{prop:half}
  Let $(\TCDF[X],\TCDF[Y])$ be a Bayes-Nash Equilibrium of the endogenous test
  selection game, where agents may be restricted to an arbitrary set.
  For every threshold $\theta$ in the support of \TCDF[X],
  $X$ must be selected with probability exactly \half when choosing $\theta$.
\end{proposition}

\begin{proof}
  Since every $\theta$ in the support of \TCDF[X] is a best response
  to \TCDF[Y], the probability of $X$ being selected when choosing
  $\theta$ does not depend on $\theta$. Denote this probability by $p_X$.
  Similarly, the probability of $Y$ being selected when choosing
  a threshold in the support of \TCDF[Y] is equal to a constant $p_Y$
  independent of the threshold chosen. Since $X$ can guarantee
  that it is selected with probability \half by ``strategy stealing''
  (i.e., sampling its threshold at random from $Y$'s equilibrium
  distribution \TCDF[Y]), the best-response condition implies
  $p_X \ge \half$, and similarly $p_Y \ge \half$. The equation
  $p_X + p_Y = 1$, reflecting the fact that exactly one firm is
  always selected, now implies $p_X = p_Y = \half$.
\end{proof}

We now define some key quantities for reasoning about the structure of the endogenous test-selection equilibria.
Consider a firm $X$ facing firm $Y$ whose test threshold is drawn from the distribution \TCDF (which may not be continuous). Let \WinPass[\TCDF]{\theta} be the probability that $X$ is selected conditioned on choosing a threshold of $\theta$ and \emph{passing} its test, and \WinFail[\TCDF]{\theta} the probability of being selected conditioned on choosing a threshold $\theta$ and \emph{failing} its test.
We define the following notation:
\begin{align*}
  \TCdfInt{\theta} & = \int_0^{\theta} \TCdf{t} \dd t &
  \FP[\TCDF]  & = \Expect[\Theta \sim \TCDF]{\Theta} \; = \; 1 - \TCdfInt{1}.
\end{align*}
\FP[\TCDF] is the \emph{failure probability under \TCDF},
i.e., the probability that a firm using the strategy \TCDF fails its test drawn from \TCDF. We will also write $\FP[X] = \FP[{\TCDF[X]}]$ and $\FP[Y] = \FP[{\TCDF[Y]}]$ for brevity.
The following lemma characterizes the probabilities of being selected.

\begin{lemma} \label{lem:basic-winning-probabilities}
Let $\PMT = \TCdf{\theta} - \lim_{t \uparrow \theta} \TCdf{t}$ be the discrete probability mass of \TCDF at $\theta$; if \TCDF is continuous at $\theta$, then $\PMT = 0$.
We have that
\begin{align*}
\WinPass[\TCDF]{\theta} & = \FP[\TCDF] + (1-\theta)\cdot \left(\TCdf{\theta} - \frac{\PMT}{2}\right) + \TCdfInt{\theta}
& \WinFail[\TCDF]{\theta} & = \theta\cdot \left(\TCdf{\theta} -\frac{\PMT}{2}\right) - \TCdfInt{\theta}.
\end{align*}
\end{lemma}

\begin{proof}
Assume that $Y$'s test is drawn from \TCDF, and $X$ chooses a test of $\theta$.
If $X$ passes its test, then it is selected in the following (disjoint)
scenarios: (1) $Y$ fails its test (which has probability \FP[\TCDF]);
(2) $Y$ passes a test of exactly $\theta$, but $X$ wins the coin flip (which has probability $\frac{(1-\theta)\PMT}{2}$);
(3) $Y$ passes a strictly easier test than $\theta$, which is an event with probability
$(\TCdf{\theta}-\PMT) \cdot \ExpectC[\Theta_Y \sim \TCDF]{1-\Theta_Y}{\Theta_Y < \theta}
= (1-\theta) \cdot (\TCdf{\theta} - \PMT) + \TCdfInt{\theta}$.

Similarly, when $X$ fails its test, it is selected whenever
(1) $Y$ fails a test of exactly $\theta$ and $X$ wins the coin flip (which has probability $\frac{\PMT \theta}{2}$),
or (2) $Y$ fails a strictly easier test than $\theta$, which is an event of probability
$(\TCdf{\theta}-\PMT) \cdot \ExpectC[\Theta_Y \sim \TCDF]{\Theta_Y}{\Theta_Y < \theta}
= \theta (\TCdf{\theta}-\PMT) - \TCdfInt{\theta}$.
\end{proof}

By combining the two cases of the preceding lemma, we obtain the following corollary:

\begin{corollary} \label{cor:agent-utility}
  Assume that $Y$'s threshold is drawn from the (possibly discontinuous) cdf \TCDF. Let $\theta$ be the threshold chosen by $X$, and $\PMT = \TCdf{\theta} - \lim_{t \uparrow \theta} \TCdf{t}$.
  Then, $X$'s probability of being selected is
\begin{align*}
  (1-\theta) \cdot \WinPass[\TCDF]{\theta}
  + \theta \cdot \WinFail[\TCDF]{\theta}
  & = (1-\theta) \cdot \FP[\TCDF] + ((1-\theta)^2 + \theta^2) \cdot \left(\TCdf{\theta} - \frac{\PMT}{2}\right)
  + (1-2\theta) \cdot \TCdfInt{\theta}.
\end{align*}
\end{corollary}

Much of the technical work of solving for equilibria of endogenous test selection games goes into showing that at equilibrium, \TCDF[X] and \TCDF[Y] are continuous and have full support.
The following lemma is proved in \dkedit{\cref{sec:characterization-proof}} as a special case of the more general treatment for agents restricted to arbitrary intervals.

\begin{lemma} \label{lem:continuity-support}
  Let $(\TCDF[X],\TCDF[Y])$ be a Bayes-Nash Equilibrium for unconstrained firms, i.e., allowed to choose tests from $[0,1]$.
  Then both \TCDF[X] and \TCDF[Y] have full support and are continuous over $[0,1]$.
\end{lemma}

%
%
%
%

\dkdelete{The high-level idea for the (fairly technical) proof of this lemma is that both point masses and gaps in the support offer profitable deviations for at least one of the firms. Some extra care must be taken due to the possible asymmetries, and to deal with the left endpoint of the allowed range.}

Using this lemma, we now return to the characterization of endogenous test selection equilibria.
The following theorem characterizes the unique Bayes-Nash Equilibrium on $[0,1]$.


\begin{theorem}
\label{thm:fullsupportequilibrium}
There is a unique Bayes-Nash Equilibrium of the endogenous test selection game when firms have access to all tests in $[0,1]$.
  The unique Bayes-Nash Equilibrium is symmetric, and its equilibrium distribution $\TCDF[X] = \TCDF[Y] = \TCDF[\text{eq}]$ has the following cdf \TCDF[\text{eq}] and pdf \TPDF[\text{eq}].

\begin{align*}
\TCDF[\text{eq}](\theta)
& = \half \cdot  \left( 1 - \frac{1-2\theta}{\sqrt{\theta^2 +
    (1-\theta)^2}} \right) 
& \TPDF[\text{eq}](\theta)
& = \half \cdot  \frac{1}{(\theta^2 + (1-\theta)^2)^{3/2}}.
\end{align*}
\end{theorem}

\begin{proof}
  We derive \TCDF[Y]; the argument for \TCDF[X] is identical.
  Recall that by \cref{lem:continuity-support}, both \TCDF[X] and \TCDF[Y] are continuous and have full support.
  First, \cref{cor:agent-utility} and \cref{prop:half}, applied to $\theta = 0$, imply that $\FP[Y] = \half$.
Let $\theta \in [0,1]$ be arbitrary.
By combining \cref{cor:agent-utility} and \cref{prop:half} with the fact that $\FP[Y] = \half$,
we obtain that  \TCDF[Y] must satisfy
\begin{align*}
(1-\theta) \cdot \left(\half + (1-\theta) \TCdf[Y]{\theta} + \TCdfInt[Y]{\theta}\right)
+ \theta \cdot \left(\theta \TCdf[Y]{\theta} - \TCdfInt[Y]{\theta}\right)
  & = \half,
\end{align*}
which we can rearrange to
$(\theta^2 + (1-\theta)^2) \cdot \TCdf[Y]{\theta} + (1-2\theta) \cdot \TCdfInt[Y]{\theta}
 = \frac{\theta}{2}$.
Dividing both sides by $(\theta^2 + (1-\theta)^2)^{3/2}$, and using the fact that the derivative of \TCDFINT[Y] is \TCDF[Y] (by definition), we obtain the differential equation
\begin{align*}
(\theta^2 + (1-\theta)^2)^{-1/2} \cdot \frac{\dd \TCdfInt[Y]{\theta}}{\dd \theta}
  + \frac{1-2\theta}{(\theta^2 + (1-\theta)^2)^{3/2}} \cdot \TCdfInt[Y]{\theta}
& = \frac{\theta}{2 (\theta^2 + (1-\theta)^2)^{3/2}}.
\end{align*}
Next, observe that the derivative of $(\theta^2 + (1-\theta)^2)^{-1/2}$ is
$\frac{1-2\theta}{(\theta^2 + (1-\theta)^2)^{3/2}}$. Using this, we can simplify the above equation (via integration by parts) to obtain that
\begin{align*}
\frac{\dd}{\dd \theta} \frac{\TCdfInt[Y]{\theta}}{\sqrt{\theta^2 + (1-\theta)^2}}
& = \frac{\theta}{2 (\theta^2 + (1-\theta)^2)^{3/2}},
\end{align*}
which, in conjunction with the boundary condition
$\TCdfInt[Y]{0} = 0$, implies
\begin{align*}
\frac{\TCdfInt[Y]{\theta}}{\sqrt{\theta^2 + (1-\theta)^2}}
  & =
    \int_0^{\theta} \frac{t}{2 (t^2 + (1-t)^2)^{3/2}} \dd t =  \half \left( 1 - \frac{1-\theta}{\sqrt{\theta^2 + (1-\theta)^2}} \right).
\end{align*}
We can simplify this equation further to get that
$\TCdfInt[Y]{\theta} = \half \sqrt{\theta^2 + (1-\theta)^2} - \frac{1-\theta}{2}$.
Next, taking a derivative, we obtain that
\begin{align*}
  \TCdf[Y]{\theta}
  & =  \frac{2\theta-1}{2\sqrt{\theta^2 + (1-\theta)^2}} + \half =
  \half \cdot  \left( 1 - \frac{1-2\theta}{\sqrt{\theta^2 + (1-\theta)^2}} \right) ,
\end{align*}
as claimed.
Since we only used that both firms' distributions have full support, the same argument can be applied to \TCDF[X], to prove that $\TCDF[X] = \TCDF[\text{eq}]$.
Finally, to verify that $\TCDF[X] = \TCDF[Y] = \TCDF[\text{eq}]$ is in fact an equilibrium, we can substitute $\TCDF=\TCDF[\text{eq}]$ into \cref{cor:agent-utility} and verify that the selection probability of the firm $X$, when faced with \TCDF[Y], is indeed exactly \half for all $\theta \in [0,1]$.
\end{proof}

\Cref{fig:unrestricted-equilbrium} shows the cdf and pdf of the equilibrium distribution of \cref{thm:fullsupportequilibrium}.
Observe that the cdf satisfies the claims established in \cref{lem:continuity-support}, namely, that it is continuous and has support $[0,1]$. Moreover, note that the pdf \TPDF[\text{eq}] is also symmetric about \half. (This is not a priori obvious, and indeed, will not be the case when we consider restricted test sets in the next section). Finally, as discussed before, observe that if quality levels $X,Y$ are drawn from any absolutely continuous distribution \PCDF, then the unique equilibrium distribution for thresholds $\sigma\in\setR$ is given by $\TCDF_{\text{eq},\PCDF}(\sigma) = \TCDF_{\text{eq}}(\PCdf{\sigma}) = \half  \left( 1 - \frac{1-2\PCdf{\sigma}}{\sqrt{\PCdf{\sigma}^2 + (1-\PCdf{\sigma})^2}} \right)$.
\begin{figure}[th]
    \centering
    \begin{subfigure}[t]{0.5\textwidth}
        \centering
        \includegraphics[width=.8\textwidth]{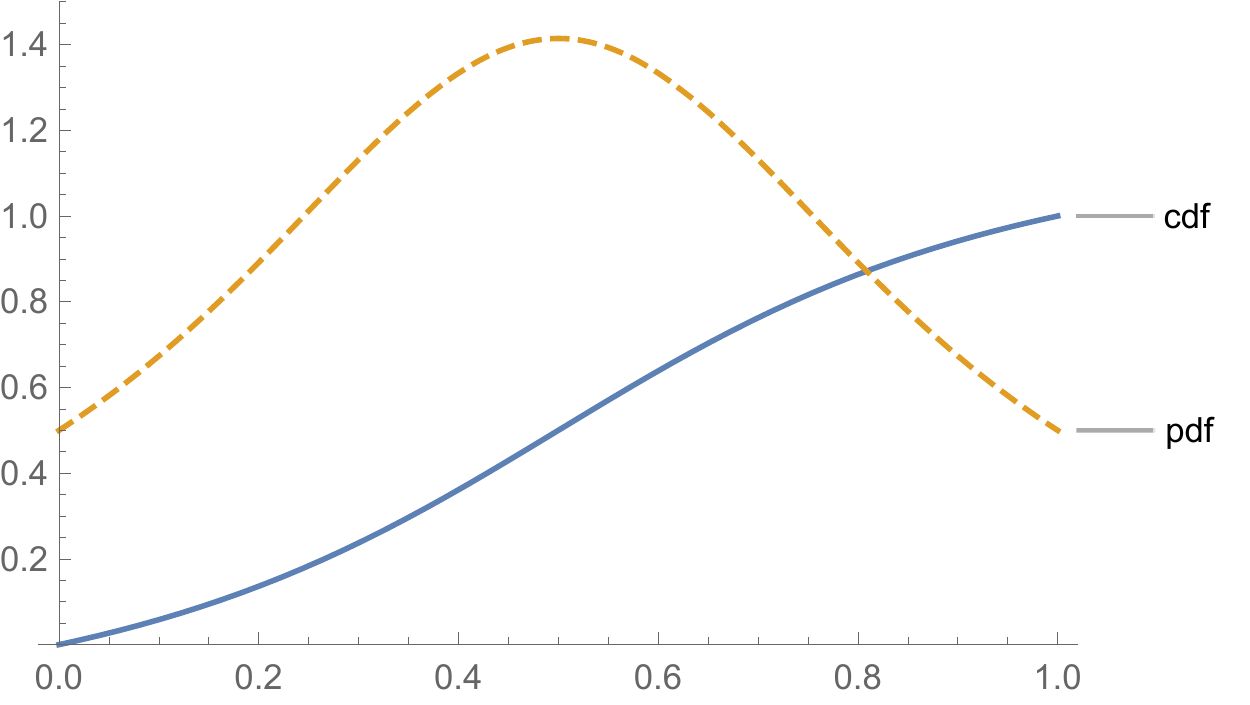}
        \caption{Equilibrium cdf and pdf for unrestricted firms.\label{fig:unrestricted-equilbrium}}
    \end{subfigure}%
    ~
    \begin{subfigure}[t]{0.5\textwidth}
        \centering
        \includegraphics[width=.8\textwidth]{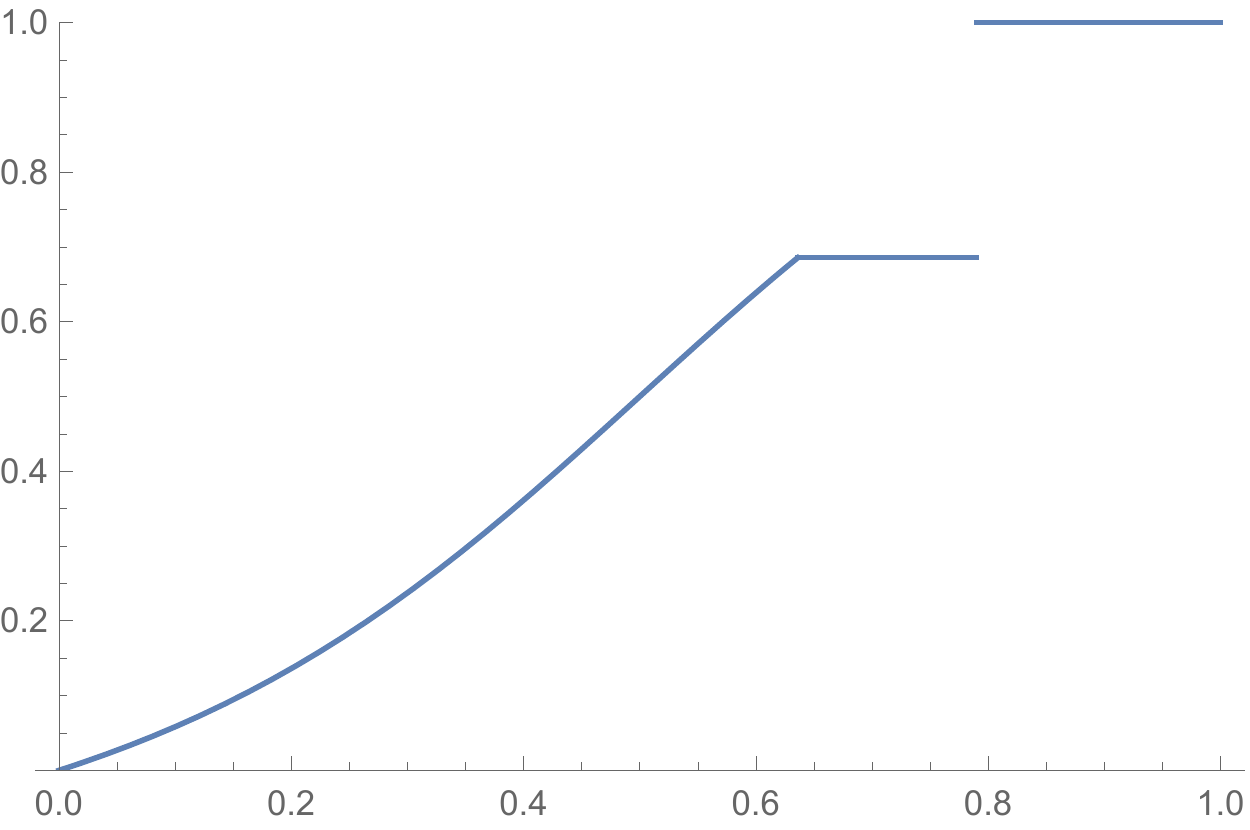}
        \caption{Equilibrium cdf for firms restricted to $[0,0.79]$.\label{fig:restricted-equilibrium}}
    \end{subfigure}
  \caption{\em Examples of equilibrium cdfs for unrestricted and restricted sets of tests. \label{fig:equilibrium}}
\end{figure}

\subsection{Price of Anarchy of (Unrestricted) Endogenous Test Selection}
\label{ssec:Poa-unrestricted-equilibrium}

We are now in a position to combine \cref{cor:optimal-cost} and \cref{thm:fullsupportequilibrium} to determine the Price of Anarchy (in terms of the principal's probability of selecting the wrong firm) of allowing firms to choose their own tests.
Substituting the characterizations into the functional (Eq.~\eqref{eqn:inversion-objective}), the resulting expression unfortunately does not lend itself to closed-form evaluation. However, a numerical calculation establishes the following.

\begin{corollary}
\label{cor:poa_unrestricted}
The equilibrium cdf \TCDF[\text{eq}] satisfies that $\Inv{\TCDF[\text{eq}]} \approx 0.23056$. Consequently, compared to the optimal i.i.d.~test selection rule, endogenous test selection over unrestricted tests has a Price of Anarchy of roughly $1.10653$ for any number of firms. Compared to the optimal correlated test selection rule, it has a Price of Anarchy of approximately $1.38336$ for two firms, decreasing to $1.10653$ as the number of firms $n \to \infty$.
\end{corollary}

\section{Endogenous Test Selection with Restricted Tests} \label{sec:equilibrium-restricted}
We now consider a more general treatment:
the principal restricts the firms to choose tests from a non-empty closed set $S \subseteq [0,1]$,
and the firms will play according to equilibrium distributions \TCDF[X], \TCDF[Y]
supported on subsets of $S$.
Note that although the firms' tests are restricted to the set $S$, their products' qualities are still drawn uniformly from the entire interval $[0,1]$; this is reflected in the probabilities of passing/failing tests.

The existence of a (mixed, symmetric) Bayes-Nash Equilibrium follows from Lemma~7 of \cite{dasgupta:maskin:theory}.
However, we note that in general, the Nash equilibrium may not be unique;
for example, when $S = \SET{1-\frac{\sqrt{2}}{2}, \frac{\sqrt{2}}{2}}$,
every pair of probability distributions on $S$ constitutes an equilibrium.
To see that this is the case,
observe that conditional on the firms choosing \emph{any}
ordered pair of tests in the product set $S \times S$,
each firm's probability of being selected is \half.

The following pair of theorems shows that by restricting the set $S$
available to the firms, even to an interval,
the principal can achieve a strictly smaller inversion probability
than under the equilibrium for $S = [0,1]$;
however, for every non-empty set $S$,
the inversion probability under every symmetric Bayes-Nash Equilibrium is
larger by some absolute constant than the one under
the optimum i.i.d.~distribution.

\begin{theorem} \label{thm:interval-equilibrium-better}
  Let \TCDF[{[0,0.79]}] be the unique\footnote{as will be established in \cref{thm:interval-equilibrium-characterization}}
  symmetric Bayes-Nash equilibrium distribution when firms choose from the interval
  $[0,0.79]$, and \TCDF[{[0,1]}] the unique and symmetric Bayes Nash equilibrium distribution for
  unrestricted firms.
  Then\footnote{Recall that we write $\Inv{\TCDF} = \Expect{\Inv{\mathcal{T}_{\TCDF}}}$.},
$\Inv{\TCDF[{[0,0.79]}]} < 0.22975 < 0.23052 < \Inv{\TCDF[{[0,1]}]}$.
\end{theorem}

\begin{theorem} \label{thm:lower-bound-general-equilibrium}
  Let $S \subseteq [0,1]$ be an arbitrary non-empty set,
  and \TCDF any symmetric Bayes-Nash equilibrium distribution
  of firms restricted to choosing tests from $S$.
  The expected probability of choosing the wrong firm under \TCDF
  is $\Inv{\TCDF} \geq \frac{5}{24} + \frac{1}{82944}$.
\end{theorem}

We emphasize that \cref{thm:lower-bound-general-equilibrium} establishes a lower bound only for \emph{symmetric} equilibria.
For general $S$, there may be asymmetric equilibria, and they may achieve error probabilities strictly smaller than $\frac{5}{24}$.
For example, as observed above, when tests are restricted to
the set $S = \SET{1-\frac{\sqrt{2}}{2}, \frac{\sqrt{2}}{2}}$,
there is an asymmetric equilibrium in which firm $X$ always
chooses $\theta_X=1-\frac{\sqrt{2}}{2}$,
$Y$ always chooses $\theta_Y=\frac{\sqrt{2}}{2}$, and the
inversion probability is
$\half (\theta_X^2 + (\theta_Y-\theta_X)^2 + (1-\theta_Y)^2)
= 3 - 2 \sqrt{2} \approx 0.17157$,
whereas $\frac{5}{24} \approx 0.2083$.

The key to proving \cref{thm:interval-equilibrium-better} is
the following complete characterization of the unique
Bayes-Nash equilibrium when $S$ is restricted to intervals,
proved in \cref{sec:characterization-proof}.

\begin{theorem} \label{thm:interval-equilibrium-characterization}
  Let $S = [a,b]$ be a non-empty interval, and consider the game when both firms are restricted to choosing tests from $S$.
  There is a unique Bayes-Nash equilibrium, which is symmetric.
  Its cdf \TCDF[\text{eq}] is given by the following:

  \begin{enumerate}
  \item If $(1-a) \cdot b \leq \half$, then \TCDF[\text{eq}]
    is a step function at $b$,
    i.e., both firms deterministically choose $b$.
  \item Otherwise, let
    \begin{align*}
    \PMB & = \frac{1-a(1-b)-b(1-a)}{(1-a)((1-b)^2+b^2)} &
    \CP & = \frac{1-a-2b+4ab-2ab^2}{1-4(1-a)b+2(1-2a)b^2}.
    \end{align*}

    The equilibrium cdf \TCDF[\text{eq}] is given by:
    \begin{align}
      \TCdf[\text{eq}]{\theta} = \begin{cases}
       \frac{1}{2(1-a)} \cdot \left( (1-2a) + \sqrt{a^2 + (1-a)^2} \cdot
             \frac{2\theta-1}{\sqrt{\theta^2 + (1-\theta)^2}} \right)
        & \text{ for } a \leq \theta < \CP \\
        1-\PMB & \text{ for } \CP \leq \theta < b \\
        1      & \text{ for } \theta = b.
      \end{cases}
      \label{eqn:general-cdf-characterization}
    \end{align}
  \end{enumerate}
\end{theorem}

\noindent An example of the equilibrium cdf (for firms restricted to interval $[0,0.79]$) is shown in \cref{fig:restricted-equilibrium}.
Using \cref{thm:interval-equilibrium-characterization},
we can now complete the proof of \cref{thm:interval-equilibrium-better}.

\begin{extraproof}{\cref{thm:interval-equilibrium-better}}
  Even for $a=0, b=1$, it appears that there is no closed-form
  solution for the value of \Inv{\TCDF[{[0,1]}]} for the equilibrium
  distribution.
  The closed-form characterization of \TCDF[{[a,b]}] allows a numerical
  evaluation for all values of $0 \leq a < \half < b \leq 1$.
  Numerically, the optimum is achieved at $a=0, b \approx 0.79$,
  where $\Inv{\TCDF[{[0,0.79]}]} \leq 0.22975$.
\end{extraproof}

\subsection{Suboptimality of All Symmetric Equilibria}

We next give the proof of \cref{thm:lower-bound-general-equilibrium}.
Recall that we use \OPT to denote the cdf of the optimal distribution,
i.e., the uniform distribution on $[\quarter,\frac{3}{4}]$.
We begin with an easy proposition, capturing that a sufficient condition for \OPT and an arbitrary cdf \GCDF to differ by at least $\eps$ at $z$ is for \GCDF to be ``sufficiently discontinuous'' at some point $\theta$.

\begin{proposition} \label{prop:point-mass-implies-difference}
  If $\GCdf{\theta} \geq \eps + \lim_{t \uparrow \theta} \GCdf{t}$, then there exists
  a $z$ with $\Abs{\Opt{z}-\GCdf{z}} \geq \frac{\eps}{2}$.
\end{proposition}

\begin{proof}
  If $\Abs{\Opt{\theta}-\GCdf{\theta}} \geq \frac{\eps}{2}$,
  then $z=\theta$ works, so assume that
  $\Abs{\Opt{\theta}-\GCdf{\theta}} < \frac{\eps}{2}$.
  Then,
  \begin{align*}
     \frac{\eps}{2} & < \Abs{\Opt{\theta}-\lim_{\rho \to 0} \GCdf{\theta-\rho}}
     \; = \; \lim_{\rho \to 0} \Abs{\Opt{\theta}-\GCdf{\theta-\rho}}
     \; \leq \; \lim_{\rho \to 0} \Abs{\Opt{\theta-\rho}-\GCdf{\theta-\rho}} + 2\rho.
  \end{align*}
  For sufficiently small $\rho$, we therefore get that
  $\Abs{\Opt{\theta-\rho}-\GCdf{\theta-\rho}} > \frac{\eps}{2}$,
  so choosing $z = \theta-\rho$ for such a small $\rho$ completes the proof.
\end{proof}

\Cref{prop:point-mass-implies-difference} is the key ingredient to proving \cref{lem:equilibrium-far-from-optimum}, which shows that symmetric equilibrium distributions deviate far from the optimal distribution,

\begin{lemma} \label{lem:equilibrium-far-from-optimum}
  Let \TCDF be the cdf of an equilibrium distribution for some non-empty closed set $S$.
  There exists a $z \in (0,1)$ with
  $\Abs{\TCdf{z}-\Opt{z}} \geq \frac{1}{24}$.
\end{lemma}

\begin{proof}
  Let $p=\frac{1}{22}$, and $\theta = \min \Set{t}{\TCdf{t} \geq 1-p}$.
  Let $\PMT = 1 - \lim_{t \uparrow \theta} \TCdf{t}$ be the point mass at $\theta$ (if any),
  and $\PMA = \TCdf{a}$ the point mass at the lower end $a$ of the support.
  Let $q$ be the probability that firm $Y$ chooses a test $y > \theta$ and fails.
  Then, $\WinPass{\theta} = (\TCdf{\theta}-\PMT) + \PMT \cdot (1-\frac{1-\theta}{2}) + q$
  and $\WinFail{\theta} = \FP - \PMT \cdot \frac{\theta}{2} - q$.
  Thus, the probability for $X$ to be chosen is
  \begin{align*}
    \half & \geq
    (1-\theta) \cdot (\TCdf{\theta} - \PMT + \PMT \cdot (1-\frac{1-\theta}{2}) + q)
    + \theta \cdot (\FP - \PMT \cdot \frac{\theta}{2} - q)
    \\ & = (1-2\theta) \cdot q + (1-\theta) \cdot \TCdf{\theta} + \theta \FP - \frac{\PMT}{2} \cdot ((1-\theta)^2+\theta^2).
  \end{align*}

  If $\theta \leq \half$, then we get that $\TCdf{\half} \geq 1-p$,
  so $\Abs{\TCdf{\half} - \Opt{\half}} \geq \half-p \geq \frac{1}{24}$.
  Otherwise, $1-2\theta < 0$, so we can lower-bound
  \begin{align*}
    (1-2\theta) \cdot q
    & \geq (1-2\theta) \cdot (1-\TCdf{\theta})
    \; \geq \; (1-2\theta) \cdot p.
  \end{align*}

  Furthermore, we lower-bound
  \begin{align*}
    \FP & = \frac{1 - \PMA(a^2+(1-a)^2)}{2(1-a)}
   \; \geq \; \frac{1-\PMA}{2}.
  \end{align*}

  Substituting these bounds, as well as $(1-\theta)^2+\theta^2 \leq 1$,
  we can lower-bound
  \begin{align*}
    \half & \geq (1-2\theta) \cdot p + (1-\theta) \cdot (1-p) + \theta \FP - \frac{\PMT}{2}
    \\ & = 1 - \theta \cdot (1+p) + \theta \cdot \frac{1-\PMA}{2} - \frac{\PMT}{2}
    \\ & \geq 1 - \theta \cdot (\half+p) - \frac{\PMA+\PMT}{2},
  \end{align*}
  so $\theta \geq \frac{1-(\PMA+\PMT)}{1+2p}$.
  If $\PMA \geq \frac{1}{12}$ or $\PMT \geq \frac{1}{12}$,
  then the lemma follows by applying \cref{prop:point-mass-implies-difference}
  with $\eps = \frac{1}{12}$ and $z=a$ or $z=\theta$.
  Otherwise, $\theta > \frac{3}{4}$, so $\Opt{\theta} = 1$,
  while $\lim_{t \uparrow \theta} \TCdf{t} \leq 1-p$.
  Therefore, there must exist a $z$ with
  $\Abs{\Opt{z}-H(z)} \geq p > \frac{1}{24}$, completing the proof.
\end{proof}

The second key lemma shows that a large deviation at even one point implies a significantly larger error probability.

\begin{lemma} \label{lem:far-from-optimum-not-good}
  Let \GCDF be any distribution such that
  $\Abs{\GCdf{z} - \Opt{z}} \geq \eps$
  for some $z \in (0,1)$ and $\eps > 0$.
  Then, $\Inv{\GCDF} \geq \Inv{\OPT} + \frac{1}{6} \eps^3$.
\end{lemma}

\begin{proof}
  Write $\GCDF[0] = \OPT$, and
  define the one-parameter family of cumulative distribution functions
  $\GCDF[t] = t \GCDF + (1-t) \cdot \GCDF[0]$ as in the proof of \cref{thm:optimal-distribution}.
  Recall that $\Inv{\GCDF[t]} = \Inv{\GCDF[0]} + A(\GCDF) \cdot t + B(\GCDF) \cdot t^2$,
  where $A(\GCDF), B(\GCDF)$ are non-negative coefficients defined in Eq.~\eqref{eq:iht-b}.
  Suppose that $\GCdf[0]{z} - \GCdf{z} \geq \eps$;
  the case $\GCdf{z} - \GCdf[0]{z} \geq \eps$ is handled symmetrically.
  Since $\GCDF[0]'(x) \leq 2$ for all $x \in [0,1],$ we have
  $\GCdf[0]{x} - \GCdf{x} \geq \eps - 2 (x-z)$ for all $x \geq z$.
  Now recall the definition of $B(\GCDF)$ in Eq.~\eqref{eq:iht-b},
  and that the integrand in the definition of
  $B(\GCDF)$ is symmetric in $x$ and $y$. Hence,
  \begin{equation} \label{eq:bh-symm}
    B(\GCDF) = \half \int_0^1 \int_0^1
    \big( \GCdf[0]{x} - \GCdf{x} + \GCdf{y} - \GCdf[0]{y} \big)^2
    \, \dd y \dd x .
  \end{equation}

  Letting $H(x) := \GCdf[0]{x} - \GCdf{x}$, we have
  \begin{align*}
	B(\GCDF) & = \half \int_0^1 \int_0^1
                   \big( H(x) - H(y) \big)^2 \, \dd y \dd x
    \\ & =
	\half \int_0^1 \int_0^1 H(x)^2 \, \dd y \dd x +
	\int_0^1 \int_0^1 H(x) H(y) \, \dd y \dd x +
         \half \int_0^1 \int_0^1 H(y)^2 \, \dd y \dd x
    \\ & =
	\int_0^1 H(x)^2 \, \dd x +
         \left( \int_0^1 H(x) \, \dd x \right)^2
    \\ & \geq
	\int_0^1 H(x)^2 \, \dd x
	\geq
	\int_z^{z + \eps/2} H(x)^2 \, \dd x
	\geq \int_z^{z + \eps/2} (\eps - 2(x-z))^2 \, \dd x
	= \half \int_{0}^{\eps}  u^2 \, \dd u
	= \frac{1}{6} \eps^3 .
    \end{align*}
    Using the inequalities $A(\GCDF) \geq 0$
    and $B(\GCDF) \geq \frac16 \eps^3$ in
    the expression $\Inv{\GCDF[t]} = \Inv{\GCDF[0]} + A(\GCDF) \cdot t + B(\GCDF) \cdot t^2$,
    and setting $t=1$ so that $\GCDF[t] = \GCDF$, we find that
    $\Inv{\GCDF} \geq \Inv{\GCDF[0]} + \frac{1}{6} \eps^3$, as claimed.
\end{proof}

Combining \cref{lem:equilibrium-far-from-optimum} and \cref{lem:far-from-optimum-not-good}, with $\eps = \frac{1}{24}$, immediately implies \cref{thm:lower-bound-general-equilibrium}.

\subsection{Proofs of \cref{thm:interval-equilibrium-characterization} and \cref{lem:continuity-support}}
\label{sec:characterization-proof}
We now characterize the Bayes-Nash equilibria of the game, in particular proving \cref{thm:interval-equilibrium-characterization} and \cref{lem:continuity-support}. As with other BNE characterizations, much of the technical work focuses on characterizing the supports of the two firms' distributions and ruling out discontinuities, except possibly at the upper end of the allowed range.


\subsubsection{Laying the Groundwork}

We begin with a lemma relating the winning probabilities conditioned on passing/failing tests and for tests of different thresholds.

\begin{lemma} \label{lem:wdiff}
  For any distribution \TCDF, for any $\theta \in [a,b]$,
  \begin{equation} \label{eq:wpfdiff}
    \WinPass[\TCDF]{\theta} - \WinFail[\TCDF]{\theta} > 0.
  \end{equation}
  Also, if $a \le \theta_0 < \theta_1 \le b$ then
  \begin{align} \label{eq:wpdiff}
    \WinPass[\TCDF]{\theta_1} - \WinPass[\TCDF]{\theta_0}
    & =
    \frac{1-\theta_0}{2} \delta_{\theta_0} -
    \frac{1-\theta_1}{2} \delta_{\theta_1} +
    \int_{\theta_0}^{\theta_1} (\TCdf{\theta}-\TCdf{\theta_0}) \, d\theta +
    (1 - \theta_1) [\TCdf{\theta_1} - \TCdf{\theta_0}] \\
    \nonumber & \le
    \frac{1-\theta_0}{2} \delta_{\theta_0} +
    \TCdf{\theta_1} - \TCdf{\theta_0}
    \\
    \label{eq:wfdiff}
    \WinFail[\TCDF]{\theta_1} - \WinFail[\TCDF]{\theta_0}
    & =
    \frac{\theta_0}{2} \delta_{\theta_0} -
    \frac{\theta_1}{2} \delta_{\theta_1} +
    \int_{\theta_0}^{\theta_1} (\TCdf{\theta_1} - \TCdf{\theta}) \, d\theta
    + \theta_0 [\TCdf{\theta_1} - \TCdf{\theta_0}] \\
    \nonumber & \le
    \frac{\theta_0}{2} \delta_{\theta_0} +
    \TCdf{\theta_1} - \TCdf{\theta_0}.
  \end{align}
\end{lemma}

\begin{proof}
  Equations~\eqref{eq:wpdiff} and~\eqref{eq:wfdiff} follow from
  \cref{lem:basic-winning-probabilities}.
  To derive Equation~\eqref{eq:wpfdiff}, i.e., to show that a firm passing its test always has \emph{strictly} higher probability of winning than failing the same test, we argue as follows.
  Consider firm $X$ choosing test $\theta$ and firm $Y$ sampling a test from \TCDF.
  The difference $ \WinPass[\TCDF]{\theta} - \WinFail[\TCDF]{\theta} $
  is the probability that one of the following occurs:
  \begin{enumerate}
    \item $Y$ samples a test $\theta' < \theta$ and passes;
    \item $Y$ samples a test $\theta' > \theta$ and fails;
    \item $Y$ samples test $\theta' = \theta$, passes the test, but loses
      the coin toss;
    \item $Y$ samples test $\theta' = \theta$, fails the test, but wins
      the coin toss.
  \end{enumerate}
  At least one of these four events has positive probability; hence
  $ \WinPass[\TCDF]{\theta} - \WinFail[\TCDF]{\theta} > 0$.
\end{proof}

Next, we prove several lemmas devoted to characterizing the support of the equilibrium distributions $(\TCDF[X], \TCDF[Y])$. We begin with a lemma that will serve the purpose of ruling out gaps in the support, except possibly at the upper end of the interval $[a,b]$.

\begin{lemma} \label{lem:halfopen}
  Let \TCDF be a probability distribution on $[a,b]$, and
  suppose that $\alpha, \beta$ satisfy
  $a \leq \alpha < \beta \leq b$ and $\TCdf{\alpha} = \TCdf{\beta}$.
  Then, there exists an $\eps > 0$ such that no test in the open
  interval $(\alpha, \beta + \eps)$ is a best response to \TCDF.
\end{lemma}
\begin{proof}
  First consider $\theta \in (\alpha, \beta]$.
  Let $\theta' \in (\alpha, \theta)$.
  The assumption that $\TCdf{\alpha} = \TCdf{\beta}$
  implies that $\TCdf{\theta} = \TCdf{\theta'}$.
  Applying Equations~\eqref{eq:wpdiff} and~\eqref{eq:wfdiff}, we find that
  $\WinPass{\theta} = \WinPass{\theta'} = \WinPass{\beta}$
  and $\WinFail{\theta} = \WinFail{\theta'} = \WinFail{\beta}$,
  i.e., the probability of winning conditioned on passing is the same at all three thresholds, and similarly for failing.
  The difference in winning probability between $\theta'$ and $\theta$ is
  \[
    (1-\theta') \WinPass{\theta'} + \theta' \WinFail{\theta'}
    - (1-\theta) \WinPass{\theta} - \theta \WinFail{\theta},
   \; = \; (\theta - \theta') \cdot (\WinPass{\beta} - \WinFail{\beta})
   \; > \; 0
  \]
  by Inequality~\eqref{eq:wpfdiff}.
  In particular, $\theta'$ is a strictly better response to \TCDF than $\theta$,
  so $\theta$ cannot be a best response to \TCDF.

Next, consider $\theta = \beta + \eps$ (for sufficiently small $\eps$),
and let $\theta' = \alpha+\eps$.
The benefit of deviating from $\theta$ to $\theta'$ is
  \begin{align*}
    \lefteqn{(1-\theta') \WinPass{\theta'} + \theta' \WinFail{\theta'}
    - (1-\theta) \WinPass{\theta} - \theta \WinFail{\theta}} \\
    & =
      (\beta - \theta') \cdot (\WinPass{\beta} - \WinFail{\beta})
    + (1-\beta) \WinPass{\beta} + \beta \WinFail{\beta}
    - (1-\theta) \WinPass{\theta} - \theta \WinFail{\theta}
    \\ & =
      (\beta - \alpha - \eps) \cdot (\WinPass{\beta} - \WinFail{\beta})
         - (1-\beta) (\WinPass{\theta} - \WinPass{\beta})
         - \beta (\WinFail{\theta} - \WinFail{\beta})
         + \eps (\WinPass{\theta} - \WinFail{\theta})
    \\& \geq
      (\beta - \alpha - \eps) \cdot (\WinPass{\beta} - \WinFail{\beta})
        - \frac{(1-\beta)^2 + \beta^2}{2} \delta_\beta
        - (\TCDF(\theta)-\TCDF(\beta))
    \\ & \geq
      (\beta-\alpha) \cdot (\WinPass{\beta} - \WinFail{\beta})
         - \eps \cdot (\WinPass{\beta} - \WinFail{\beta})
         - (\TCDF(\theta)-\TCDF(\beta))
  \end{align*}
  where the penultimate line uses \cref{lem:wdiff}
  and the fact that $\WinPass{\theta} - \WinFail{\theta} > 0$,
  and the last line uses the observation that $\delta_\beta=0$,
  which follows from the assumption that $\TCDF(\alpha)=\TCDF(\beta)$.
  The first term on the last line is strictly positive, whereas all
  the other terms on the last line converge to zero as $\eps \to 0$.
  (Recall that cumulative distribution functions such as \TCDF are right-continuous.)
  Therefore, as $\eps \to 0$, the quantity on the last line is positive,
  so there exists some $\eps$ such that for all $\theta \leq \beta+\eps$,
  playing $\theta$ is not a best response.
\end{proof}

The next lemma shows that Bayes-Nash Equilibrium distributions can have point masses (i.e., discontinuities) at most at the upper and lower end of the allowed ranges. We will later also rule out point masses at the lower end.

\begin{lemma} \label{lem:pointmass}
  Let $(\TCDF[X], \TCDF[Y])$ be a Bayes-Nash Equilibrium for firms
  constrained to tests from $[a,b]$.
  The distributions \TCDF[X] and \TCDF[Y] have no point masses other than
  possibly at $a$ or $b$, and at most one of them has a point mass at $a$.
\end{lemma}

\begin{proof}
  We first show that for each $\theta < b$, at most one firm has point mass at $\theta$.
  Denote \TCDF[Y] by \TCDF, and let $\theta < b$ be a point where firm $Y$ has point mass $\PMT > 0$.
  We now compare the winning probability for a firm $X$ with threshold $\theta+\eps$ to the winning probabilities for $X$ with thresholds $\theta$ or $\theta-\eps$.
  First, using \cref{lem:wdiff} and straightforward calculations,
  for every $\eps > 0$, we can bound
  \begin{align*}
    \WinPass{\theta+\eps} \geq \WinPass{\theta} + \frac{\PMT}{2} \cdot (1-\theta)
    & \qquad
    \WinFail{\theta+\eps}  \geq \WinFail{\theta} + \frac{\PMT}{2} \cdot \theta  \\
    \WinPass{\theta-\eps}  \leq \WinPass{\theta} - \frac{\PMT}{2} \cdot (1-\theta)
    & \qquad
    \WinFail{\theta-\eps}  \leq \WinFail{\theta} - \frac{\PMT}{2} \cdot \theta.
  \end{align*}

  Therefore, the winning probability of firm $X$ with threshold $\theta+\eps$ is
  \begin{align}
      \nonumber
    (1-\theta-\eps) \WinPass{\theta+\eps} +
    (\theta+\eps) \WinFail{\theta+\eps}
    & \geq
    (1-\theta-\eps) \left( \WinPass{\theta} + \frac{\PMT}{2} \cdot (1-\theta) \right)
    + (\theta+\eps) \left( \WinFail{\theta} + \frac{\PMT}{2} \cdot \theta \right) \\
      \label{eq:compare-theta}
    & \geq (1-\theta) \WinPass{\theta} + \theta \WinFail{\theta}
      + \frac{\PMT}{2} \cdot ((1-\theta)^2 + \theta^2)
      - \eps \\
      \nonumber
    & \geq (1 - \theta + \eps) \left( \WinPass{\theta} -
    \frac{\PMT}{2} \cdot (1-\theta) \right)
      + (\theta-\eps) \left( \WinFail{\theta} - \frac{\PMT}{2}
      \cdot \theta \right) \\
      \nonumber
    & \quad  + \PMT \cdot ((1-\theta)^2 + \theta^2)
      - 2 \eps \\
      \label{eq:compare-theta-minus-eps}
    & \geq (1 - \theta + \eps) \WinPass{\theta-\eps} +
      (\theta-\eps) \WinFail{\theta-\eps}
      + \PMT \cdot ((1-\theta)^2 + \theta^2) - 2 \eps.
  \end{align}
  For sufficiently small positive $\eps$, the quantity
  $\frac{\PMT}{2} \cdot ((1-\theta)^2 + \theta^2) - \eps$
  is strictly positive,
  so neither $\theta$ nor $\theta-\eps$
  can be a best response for firm $X$.
  Therefore, in a Bayes-Nash Equilibrium, the probability
  of $X$ choosing a test in the set $[\theta-\eps,\theta]$
  is zero.

  We have shown that if one firm has a point mass
  at $\theta$, then the other does not.
  When $\theta=a$, this is all the lemma requires us to prove.
  When $\theta \in (a, b)$, we need to show that {\em neither}
  of the distributions \TCDF[X], \TCDF[Y] can have
  a point mass at $\theta$.
  For the sake of contradiction, assume that \TCDF[Y] has a
  point mass at $\theta \in (a,b)$.
  Let $\eps \in (0, \theta-a)$ be small enough that the probability
  of $X$ choosing a test in $[\theta-\eps,\theta]$
  is zero; such an $\eps$ exists by the preceding argument.
  It follows that $\TCDF[X](\theta-\eps) = \TCDF[X](\theta)$.
  Now, using \cref{lem:halfopen},
  we may conclude that for some $\eps' > 0$, no test in
  the interval $(\theta-\eps,\theta+\eps')$ is a
  best response to $\TCDF[X]$. Hence \TCDF[Y], which
  has a point-mass at $\theta$, cannot be a best response
  to \TCDF[X], in contradiction to our assumption that
  \TCDF[X] and \TCDF[Y] constitute a Bayes-Nash Equilibrium.
\end{proof}

The next lemma pins down the support of equilibrium distributions, showing that both firms' distributions have the same support, and showing that it must be of a very specific form.

\begin{lemma} \label{lem:eq-supt}
  Let $(\TCDF[X], \TCDF[Y])$ be a Bayes-Nash Equilibrium for firms
  constrained to tests from $[a,b]$. The probability distributions
  \TCDF[X] and \TCDF[Y] have the same support.
  This support set is one of the following three alternatives.
  \begin{itemize}
    \item An interval $[a,\CP]$ where $a < \CP \le b$.
    \item A set of the form $[a,\CP] \cup \{b\}$ where
    $a \le \CP < b$.
    \item The set $\{b\}$.
  \end{itemize}
\end{lemma}

\begin{proof}
  Denote the complements of the support
  sets of \TCDF[X],\TCDF[Y] by $U_X,U_Y$, respectively.
  Both of these sets are open, since the support
  of a distribution is, by definition, closed.
  If $U_X$ and $U_Y$ are both empty, the lemma's
  conclusion is satisfied.
  Otherwise, assume without loss of generality that $U_X$ is non-empty.
  Consider an arbitrary $\theta \in U_X$, and let
  $J = (\alpha,\beta)$ be the maximal open subinterval of
  $U_X$  containing $\theta$.
  (Here, we also consider a half-open interval of the form
  $[a,\beta)$ or $(\alpha,b]$ to be an open subinterval of $U_X$.)
  For any $\beta' \in J$ we have
  $\TCdf[X]{\alpha} = \TCdf[X]{\beta'}$ which
  implies, by \cref{lem:halfopen},
  that none of the points in $(\alpha,\beta')$ is a best response
  to \TCDF[X]. Therefore, none of these points
  is in the support of \TCDF[Y], i.e., the interval
  $(\alpha,\beta')$ is contained in $U_Y$.
  Taking the union over all $\beta' \in J$, we find that
  $J  = \bigcup_{\beta' \in I} (\alpha,\beta') \subseteq U_Y$;
  in particular, this means that $\theta \in U_Y$.
  As $\theta \in U_X$ was arbitrary, we have shown that
  $U_X \subseteq U_Y$. A symmetric argument establishes that
  $U_Y \subseteq U_X$, so the two sets are equal, and their
  complements, the support sets of \TCDF[X] and \TCDF[Y],
  are equal as well.

  We now turn to proving that the support sets have one
  of the three structures enumerated in the statement of the lemma.
  Equivalently, we will prove that if the complementary set
  $U = U_X = U_Y$ is non-empty, then it is an interval of the form
  $[a,b)$, $(\CP,b)$, or $(\CP,b]$, where $\CP > a$.
  To prove this, let us return to reasoning about
  the interval $J = (\alpha,\beta)$,
  a maximal open subinterval of $U$;
  recall that we consider the interval $(\alpha, b]$ open.
  For the sake of contradiction, assume that $\beta < b$.
  \Cref{lem:pointmass} implies that \TCDF[X] has no point mass
  at $\beta$, so $\TCdf[X]{\alpha} = \TCdf[X]{\beta}$.
  Therefore, \cref{lem:halfopen} implies that
  for some $\eps>0$, none of the points in
  $(\alpha,\beta+\eps)$ is a best response to \TCDF[X].
  Therefore, none of these points is a support point of
  \TCDF[Y]. This contradicts the facts that $U_Y=U$
  and that $J$ is a maximal open subinterval of $U_Y$.

  We have shown that every maximal open subinterval
  of $U$ has right endpoint $b$. If the left endpoint
  is $\CP > a$ then $U_X$ has the form $(\CP,b)$ or $(\CP,b]$,
  which exactly corresponds to the first two types in the statement of the lemma.
  If the left endpoint of $U$ is $a$, then $U$ is one of the
  sets $(a,b)$, $(a,b]$, or $[a,b)$.
  The first two alternatives can be eliminated because they both imply that
  \TCDF[X] has an isolated support point at $a$.
  Since both distributions have the same support,
  this means that \TCDF[Y] also has an isolated support
  point at $a$. However, an isolated point in the
  support of a distribution must be a point mass,
  and \cref{lem:pointmass} guarantees that at most one of \TCDF[X], \TCDF[Y] has
  a point mass at $a$.
\end{proof}

\begin{lemma} \label{lem:fppmb}
  If $(\TCDF[X], \TCDF[Y])$ is a Bayes-Nash Equilibrium,
  then both distributions have $b$ in their support,
  and neither of them has a point mass at $a$.
  If $(1-a) \cdot b \le \half$ then the only equilibrium
  is a step function at $b$, i.e., both firms deterministically choose $b$.
  Otherwise, we have
  the following characterizations of the failure
  probability, the point mass at $b$,
  and the (common) largest support point other than $b$, for both firms:
  \begin{align}
    \FP[X] \; = \; \FP[Y] \; = \; \FP & = \frac{1}{2(1-a)}
      \label{eqn:failure-probability-general} \\
    \PMBX \; = \; \PMBY \; = \; \PMB & = \frac{1-2b+2b \FP }{(1-b)^2+b^2}
          \label{eqn:pointmass-general} \\
    \CP & = \frac{1-a-2b+4ab-2ab^2}{1-4(1-a)b+2(1-2a)b^2}.
  \end{align}
\end{lemma}

\begin{proof}
  If the equilibrium is not a step function at $b$,
  then \cref{lem:eq-supt} implies
  that the (common) support of \TCDF[X] and \TCDF[Y]
  contains an interval $[a,\CP]$, with $\CP>a$.
  Consider firm $Y$ playing $a + \eps$, where $\eps < \CP - a$.
  \Cref{cor:agent-utility} with $\theta=a+\eps$, where $\PMT=0$,
  characterizes the probability of $Y$ being chosen as
  \[
(1 - a - \eps) \cdot \FP[X] + ((1-a-\eps)^2 + (a+\eps)^2)
\cdot \TCdf[X]{a+\eps} + (1 - 2 a - 2 \eps) \cdot \TCdfInt[X]{a+\eps}.
  \]
  This must equal \half by \cref{prop:half}.
  Taking the limit as $\eps \to 0$, and observing
  that $\lim_{\eps \to 0} \TCdf[X]{a+\eps} = \TCdf[X]{a}$
  and $\lim_{\eps \to 0} \TCdfInt[X]{a+\eps} = 0$,
  we obtain that
  \begin{align} \label{eq:fppmt.0}
    \FP[\dkedit{X}] & = \frac{1 - \TCdf[X]{a} \cdot ((1-a)^2 + a^2)}{2(1-a)}.
  \end{align}
  A symmetric derivation, with the roles of $X$ and $Y$ reversed, establishes that $\FP[\dkedit{Y}] = \frac{1 - \TCdf[Y]{a} \cdot ((1-a)^2 + a^2)}{2(1-a)}$.
  Notice that these expressions are equal to the expression for \FP
  in Equation~\eqref{eqn:pointmass-general},
  provided that $\TCdf[X]{a} = \TCdf[Y]{a} = 0$,
  i.e.~\TCDF[X] and \TCDF[Y] have no point mass at $a$.
  We will prove this fact below.

  By \cref{lem:pointmass}, at most one firm has a point mass at $a$;
  assume without loss of generality that firm $X$ has no point mass at $a$,
  so $\TCdf[X]{a} = 0$,
  whence the failure probability of firm $X$
  is $\FP[X] = \frac{1}{2(1-a)}$. If $(1-a)b \le \half$
  this means that $\FP[X] \ge b$. However, the only way
  that the failure probability of a firm sampling a test
  from $[a,b]$ could be as large as $b$ is if $X$ deterministically
  chooses a test of difficulty $b$.
  Since \TCDF[X] and \TCDF[Y] have the same support, this
  means that $Y$ also deterministically chooses $b$,
  i.e., we have confirmed that when $(1-a)b \le \half$,
  the only equilibrium is that both firms deterministically choose $b$.

  Assume henceforth that $(1-a) b > \half$ and let
  $\theta$ be the supremum of the supports of \TCDF[X] and \TCDF[Y].
  We shall prove that $\theta=b$, as claimed by the lemma,
  while also establishing Equation~\eqref{eqn:pointmass-general}.
  Consider firm $X$ or $Y$ playing $\theta$, and apply
  \cref{cor:agent-utility} and \cref{prop:half} to obtain
  \begin{align} \label{eq:fppmb.1}
  \begin{split}
    \half & = (1-\theta) \cdot \FP[Y]
    + ((1-\theta)^2 + \theta^2) \cdot (\TCdf[Y]{\theta} - \half \PM[\theta,Y])
    + (1 - 2\theta) \cdot \TCdfInt[Y]{\theta} \\
    \half & = (1-\theta) \cdot \FP[X]
    + ((1-\theta)^2 + \theta^2) \cdot (\TCdf[X]{\theta} - \half \PM[\theta,X])
    + (1 - 2\theta) \cdot \TCdfInt[X]{\theta}.
  \end{split}
  \end{align}
  By our choice of $\theta$,
  we have that $\TCdf[Y]{\theta} = \TCdf[X]{\theta} = 1$, as well as
  $\TCdfInt[Y]{\theta} = \TCdfInt[Y]{1} - (1-\theta) = \theta - \FP[Y]$ and
  $\TCdfInt[X]{\theta} = \TCdfInt[X]{1} - (1-\theta) = \theta - \FP[X]$.
  Substituting these into the right-hand sides of Equation~\eqref{eq:fppmb.1},
  we obtain that
  \begin{align}
     \half & = (1 - \theta + \theta \FP[Y]) - \PM[\theta,Y] \cdot \frac{(1-\theta)^2+\theta^2}{2} \label{eq:fppmb.2} \\
     \half & = (1 - \theta + \theta \FP[X]) - \PM[\theta,X] \cdot \frac{(1-\theta)^2+\theta^2}{2} \label{eq:fppmb.3}.
  \end{align}

  Recall that we are assuming without loss of generality
  that $\TCdf[X]{a}=0$ and
  $\FP[X] = \frac{1}{2(1-a)}$.
  If $\PM[\theta,X]=0$, we can rearrange Equation~\eqref{eq:fppmb.3}
  to obtain $\theta (1 - \FP[X]) = \half$.
  Since $\theta \leq 1$ and
  $1 - \FP[X] = 1 - \frac{1}{2(1-a)} \leq \half$,
  the only way this  equation could hold is if $a=0$ and $\theta=1$.
  Hence, either \TCDF[X] has a point mass at $\theta$, or $\theta=1$.
  In the former case, \cref{lem:pointmass} implies that $\theta=b$.
  In the latter case, $\theta = 1 = b$.
  Thus, in either case, we have proved that $\theta = b$.

  Substituting $\theta=b$ into Equations~\eqref{eq:fppmb.2} and \eqref{eq:fppmb.3}
  and rearranging, we obtain (essentially) the expressions
  for \PMBX and \PMBY in Equation~\eqref{eqn:pointmass-general};
  more specifically, it only remains to show that $\FP[X] = \FP[Y]$,
  which will follow once we have shown that $\TCdf[Y]{a} = 0$ below.


  Now, suppose that \TCDF[X] and \TCDF[Y] are supported on
  the set $[a,\CP] \cup \SET{b}$, for some $\CP \in (a,b]$.
  We turn to calculating \CP.
  Let $\FP[X], \PMBX$ denote the failure probability
  and point mass at $b$ for firm $X$.
  Consider firm $Y$ playing \CP.
  Because $\WinPass{\CP} = 1 - (1-b) \PMBX$
  (if $Y$ passes, it will be chosen unless $X$ plays $b$ and passes),
  and $\WinFail{\CP} = \FP[\dkedit{X}] - b \PMBX$
  (if $Y$ fails, it will be chosen if $X$ fails, but did not play $b$),
  $Y$'s probability of being selected is
  \begin{align} \label{eq:fppmb.1.5}
    \half & = (1-\CP) \cdot (1 - (1-b) \PMBX)
            + \CP \cdot (\FP[X] - b \PMBX)
          \; = \; 1 - (1-b) \PM[b,X] - \CP \cdot (1+(2b-1) \PM[b,X] - \FP[X]).
  \end{align}
  The same reasoning with the roles of $X$ and $Y$ reversed implies
  \begin{align} \label{eq:fppmb.1.6}
    \half & = (1-\CP) \cdot (1 - (1-b) \PMBY) + \CP \cdot (\FP[Y] - b \PMBY)
          \; = \; 1 - (1-b) \PMBY - \CP \cdot (1+(2b-1) \PMBY - \FP[Y]).
  \end{align}
  Solving Equation~\eqref{eq:fppmb.1.6} for \CP and substituting the characterizations
  $\FP[X] = \frac{1}{2(1-a)}$ and
  $\PMBX = \frac{1-2b+2b \FP[X] }{(1-b)^2+b^2}$
  derived above gives us that
  \begin{align*}
    \CP
    & = \frac{\half-(1-b) \PMBX}{1+(2b-1) \PMBX - \FP[X]}
    \; = \; \frac{1-a-2b+4ab-2ab^2}{1-4(1-a)b+2(1-2a)b^2},
  \end{align*}
  as asserted in the statement of the lemma.

  Finally, we must prove that firm $Y$ has no point mass at $a$,
  which will imply that the derivations of \FP[X] and \PMBX apply equally to \FP[Y] and \PMBY.
  We may rearrange Equations~\eqref{eq:fppmb.1.5}--\eqref{eq:fppmb.1.6}
  to derive
  \begin{align*}
    ((1-\CP)(1-b) + \CP b) \cdot \PMBX & = \half - \CP + \CP \FP[X] \\
    ((1-\CP)(1-b) + \CP b) \cdot \PMBY & = \half - \CP + \CP \FP[Y].
  \end{align*}
  By subtracting the second equation from the first, we obtain that
  \begin{align*}
    ((1-\CP)(1-b) + \CP b) \cdot (\PMBX-\PMBY) & = \CP \cdot (\FP[X] - \FP[Y]).
  \end{align*}

  Substituting the expressions for $\PMBX, \PMBY$ derived above,
  we obtain the equation
  \begin{align} \label{eq:fppmb.4}
    \frac{((1-\CP)(1-b) + \CP b) \cdot 2b}{(1-b)^2 + b^2} \cdot (\FP[X] - \FP[Y])
    & = \CP (\FP[X] - \FP[Y]) .
  \end{align}
  Consequently, either $\FP[X]=\FP[Y]$ or
  $\frac{((1-\CP)(1-b) + \CP b) \cdot 2b}{(1-b)^2 + b^2} = \CP$.
  The latter equation can be rearranged to
  \begin{align*}
    2b - 2b^2 & = (1 - 2b^2) \CP.
  \end{align*}
  Recall that we are assuming here that $(1-a) b > \half$,
  which in particular implies that $2b > 1$, so the equation
  $2b - 2b^2 = (1 - 2b^2) \CP$
  implies that $\CP > 1$, contradicting the fact that $\CP \leq b$.
  Consequently, the equation
  $\frac{((1-\CP)(1-b) + \CP b) \cdot (2b)}{(1-b)^2 + b^2} = \CP$
  cannot be satisfied, meaning that Equation~\eqref{eq:fppmb.4}
  implies $\FP[X] = \FP[Y]$.
  Recalling the formula for failure probability in
  Equation~\eqref{eqn:failure-probability-general},
  we see that the equation $\FP[X]=\FP[Y]$
  implies that $\TCDF[Y](a) = \TCDF[X](a) = 0$,
  i.e., neither distribution has a point mass at $a$,
  as claimed.
\end{proof}

\subsubsection{Proofs of \cref{lem:continuity-support} and \cref{thm:interval-equilibrium-characterization}}

\Cref{lem:continuity-support} follows easily as a corollary of \cref{lem:pointmass} and \cref{lem:fppmb}:

\dkdeletecomment{Not sure if we should restate the results here, since they're not in an appendix.}{
\begin{rtheorem}{Lemma}{\ref{lem:continuity-support}}
  Let $(\TCDF[X],\TCDF[Y])$ be a Bayes-Nash Equilibrium for unconstrained firms, i.e., firms allowed to choose tests from $[0,1]$.
  Then both \TCDF[X] and \TCDF[Y] are continuous over all of $[0,1]$, and have full support $[0,1]$
\end{rtheorem}}

\begin{extraproof}{\cref{lem:continuity-support}}
  \Cref{lem:fppmb}, applied with $a=0$ and $b=1$, implies that \TCDF[X] and \TCDF[Y] have no point mass at $a=0$ or $b=1$ and that $\CP = 1$. Therefore, the distributions have full support. By \cref{lem:pointmass}, the distributions have no discontinuities on $(a,b)$, either, completing the proof.
\end{extraproof}

Finally, we use the characterization of the equilibrium supports to prove \cref{thm:interval-equilibrium-characterization}.

\dkdeletecomment{Not sure if we should restate the results here, since they're not in an appendix.}{
\begin{rtheorem}{Theorem}{\ref{thm:interval-equilibrium-characterization}}
  Let $S = [a,b]$ be a non-empty interval, and consider the game when both firms are restricted to choosing tests from $S$.
  There is a unique Bayes-Nash equilibrium, which is symmetric.
  Its cdf \TCDF[\text{eq}] is given by the following:

  \begin{enumerate}
  \item If $(1-a) \cdot b \le \half$, then \TCDF[\text{eq}]
    is a step function at $b$,
    i.e., both firms deterministically choose $b$.
  \item Otherwise, let
    \begin{align*}
    \PMB & = \frac{1-a(1-b)-b(1-a)}{(1-a)((1-b)^2+b^2)} &
    \CP & = \frac{1-a-2b+4ab-2ab^2}{1-4(1-a)b+2(1-2a)b^2}.
    \end{align*}

    The equilibrium cdf \TCDF[\text{eq}] is given by:
    \begin{align*}
      \TCdf[\text{eq}]{\theta} = \begin{cases}
       \frac{1}{2(1-a)} \cdot \left( (1-2a) + \sqrt{a^2 + (1-a)^2} \cdot
             \frac{2\theta-1}{\sqrt{\theta^2 + (1-\theta)^2}} \right)
        & \text{ for } a \leq \theta < \CP \\
        1-\PMB & \text{ for } \CP \leq \theta < b \\
        1      & \text{ for } \theta = b.
      \end{cases}
    \end{align*}
    \end{enumerate}
\end{rtheorem}
}

\begin{extraproof}{\cref{thm:interval-equilibrium-characterization}}
  The first case (when $(1-a) \cdot b \le \half$ is explicitly covered by \cref{lem:fppmb}. Therefore, assume from now on that $(1-a) \cdot b \geq \half$.

  By \cref{lem:eq-supt} and \cref{lem:pointmass}, we know
  that \TCDF[X] and \TCDF[Y] are both continuous over $[a,b)$ and
  that there is some $\CP \in (a,b]$ such that both functions are
  strictly monotone over $[a,\CP)$, constant over $[\CP,b)$, and then
  possibly discontinuous at $b$.
  Note that the characterization of \PMB and \CP from
  \cref{lem:fppmb} exactly correspond to the definitions of
  \PMB and \CP in \cref{thm:interval-equilibrium-characterization}.

  %

We can now generalize our approach from \cref{sec:equilibrium}
for computing the equilibrium distribution on the interval $[a,\CP]$.
Let \TCDF refer to either of the equilibrium distributions
\TCDF[X], \TCDF[Y].
Consider any threshold $\theta \in (a,\CP)$.
By \cref{prop:half} and \cref{cor:agent-utility},
\begin{align*}
  \half
   & = (1-\theta) \cdot \FP + (1-2\theta) \TCdfInt{\theta}
     + ((1-\theta)^2+\theta^2) \cdot \TCdf{\theta};
\end{align*}
compared to the derivation for $[0,1]$,
we now do not have that $\FP = \half$.
Rearranging and dividing the equation by
$(\theta^2 + (1-\theta)^2)^{3/2}$,
we obtain the differential equation

\begin{align*}
  (\theta^2 + (1-\theta)^2)^{-1/2} \cdot \TCdf{\theta}
  + \frac{1-2\theta}{(\theta^2 + (1-\theta)^2)^{3/2}} \cdot \TCdfInt{\theta}
& = \frac{\half - (1-\theta) \cdot \FP}{(\theta^2 + (1-\theta)^2)^{3/2}}.
\end{align*}

As before, we can integrate by parts to obtain that for all
$\theta \in (\alpha,\CP)$,
\begin{align*}
\frac{\dd}{\dd \theta} \frac{\TCdfInt{\theta}}{\sqrt{\theta^2 + (1-\theta)^2}}
& = \frac{\half - (1-\theta) \cdot \FP}{(\theta^2 + (1-\theta)^2)^{3/2}},
\end{align*}

so, using the boundary condition $\TCdfInt{a}=0$,
we find that for all $\theta \in [a,\CP]$,

\begin{align*}
\frac{\TCdfInt{\theta}}{\sqrt{\theta^2 + (1-\theta)^2}}
  & =
    \int_a^{\theta}  \frac{\half - (1-t) \cdot \FP}{(t^2 + (1-t)^2)^{3/2}} \dd t
\\ & = \FP \cdot \int_a^{\theta} \frac{t-a}{(t^2 + (1-t)^2)^{3/2}} \dd t
\\ & = \FP \cdot \left[ (1-a) +
     \frac{t+a-1-2 a t}{\sqrt{t^2+(1-t)^2}} \right]_a^\theta
\\ & = \FP \cdot \left(
     \frac{\theta+a-1-2a\theta}{\sqrt{\theta^2 + (1-\theta)^2}}
     - \frac{2a-1-2a^2}{\sqrt{a^2+(1-a)^2}} \right)
  \\ & = \FP \cdot \left(
   \sqrt{a^2+(1-a)^2}
   - \frac{1+2a\theta-a-\theta}{\sqrt{\theta^2 + (1-\theta)^2}}
    \right),
\end{align*}

or

\begin{align*}
  \TCdfInt{\theta}
  & = \FP \cdot \left(
    \sqrt{a^2 + (1-a)^2} \cdot \sqrt{\theta^2 + (1-\theta)^2}
    - (1+2a\theta-a-\theta) \right).
\end{align*}

Taking a derivative, we obtain that

\begin{align} \label{eq:equilibrium-general-formula}
  \TCdf{\theta}
  & = \FP \cdot \left( \sqrt{a^2 + (1-a)^2} \cdot
    \frac{2\theta-1}{\sqrt{\theta^2 + (1-\theta)^2}} + (1-2a) \right).
\end{align}
As \TCDF in this proof was interpreted to be either of
the equilibrium distributions \TCDF[X], \TCDF[Y], we
find that the only Bayes-Nash Equilibrium is the
symmetric equilibrium in which both firms' distributions
obey the formula~\eqref{eq:equilibrium-general-formula}
for all $\theta \in [\alpha,\CP)$,
and both distribution have all of their
remaining probability mass located at $b$.
\end{extraproof}


\section{Conclusions} \label{sec:conclusions}
We introduced and studied a problem of optimal and endogenous test selection in a setting where a principal wants to select the product of higher quality from one of two firms, but the products' qualities can only be measured through \emph{threshold tests} which reveal whether a product's quality lies above or below a threshold $\theta$. We explicitly characterized the optimal correlated and i.i.d.~distributions for the principal, as well as the equilibrium distribution when the firms can choose their own thresholds from an interval $[a,b]$ (in particular including the case of the interval $[0,1]$). Using these characterizations, we showed that the principal can do strictly better by giving the firms different tests than drawing their tests i.i.d. The best i.i.d.~distribution is better than any symmetric equilibrium for any set $S$ offered to the firms (including sets $S$ that are not intervals), and the equilibrium under the \emph{best interval} gives the principal strictly higher probability of selecting the best product than the equilibrium for the interval $[0,1]$.

Our work raises a wealth of questions for future work. An immediate question implicitly raised in \cref{sec:equilibrium-restricted} is which set of tests a principal should offer to achieve the smallest probability of selecting the wrong product at equilibrium.\footnote{Of course, if the principal can choose different sets for different firms, then she can choose $S_X = \{\third\}$ and $S_Y = \{\frac{2}{3}\}$, which would implement the optimal strategy for her. The more interesting question is to find one set $S$ to restrict \emph{all} firms to, which naturally corresponds to prescribing standards for quality control.} There are two variants to this question: when the principal is interested only in symmetric equilibria, or also in asymmetric (non-unique) ones. For the former version, a natural conjecture would be that the optimal set for the principal is an interval, in which case our numerical calculations from \cref{sec:equilibrium-restricted} would imply that the optimum set would be the interval $[0,0.79\ldots]$. While we cannot prove or disprove this conjecture at this point, a similar-looking stronger conjecture is false: there are discrete sets $S$ the principal can offer under which the unique equilibrium is strictly better than if the principal instead offered the smallest interval containing all of $S$. For the latter case, we conjecture optimality of the set $\SET{1-\sqrt{2}/2, \sqrt{2}/2}$, discussed in \cref{sec:equilibrium-restricted}.

The endogenous test selection game between the firms can be viewed as a natural instance of a \emph{signaling game}, in which each firm's strategy is a signaling scheme. Our problem setup severely restricted the signaling schemes the firms could choose from, to binary threshold tests. Naturally, it would be desirable to extend the results to broader classes of signaling schemes. At the full extreme, when firms may choose any signaling scheme, the unique equilibrium of our game is full disclosure. This follows from Corollary~1 of \citep{hwang:kim:boleslavksy:competitive-advertising}. However, an analysis of the intermediate regime, in which the number of signals is still constrained (as in \citep{LimitedSignaling}), would still be of interest.

Perhaps the most immediate next step along these lines would be signaling schemes in which firms can choose an arbitrary mapping from qualities to $\SET{\text{pass},\text{fail}}$. It is not hard to show that w.l.o.g., it suffices to consider signaling schemes with two thresholds $0 \leq \theta_1 \leq \theta_2 \leq 1$ in which the firm passes the test iff its quality lies in $[\theta_1, \theta_2]$. A natural conjecture would be that even if the firms were allowed to choose such tests, at equilibrium, they would always choose threshold tests only, i.e., set $\theta_1 = 0$. This conjecture is false! If such a symmetric equilibrium existed, it would have to be the equilibrium we derived in \cref{sec:equilibrium} --- however, against this strategy, there are responses yielding firm $X$ a selection probability strictly larger than \half. Explicitly characterizing the equilibrium distribution appears difficult. 

Another natural version is to require threshold tests, but allow multiple thresholds $\theta_1 \leq \theta_2 \leq \cdots \leq \theta_k$. This naturally corresponds to the type of tests encountered in classes, where cutoffs are defined between multiple grades. Even for two thresholds, characterizing the equilibrium outcomes appears difficult -- a firm with a difficult-to-attain `B' grade may have to be ranked ahead of a firm with an easy-to-attain `A' grade (similarly between easy `B' and difficult `C'\ldots). This is different from the pass-fail model, where every firm that passes a test is ranked ahead of every firm that fails a test, regardless of the tests' difficulties.


A very interesting direction for future work is considering firms whose product qualities are drawn independently from \emph{different} distributions. If one distribution stochastically dominates the other, it would be interesting to see if the weaker firm may at equilibrium follow ``moon shot'' strategies of taking very hard tests and hoping that this will allow it to win some of the time. Characterizing the equilibrium again appears to be quite challenging, because when both firms pass tests of the same difficulty, their posterior quality distributions will be different --- as a result, the principal will not simply rank passing firms by their thresholds, and this results in a possibly infinite-dimensional system of differential equations characterizing the equilibrium distribution.

\dkcomment{Bringing back these two paragraphs.}
There are several open directions in terms of alternate objectives when extending the model to $n > 2$ firms. When the principal's goal is to obtain a complete ranking minimizing the Kendall tau distance, and the firms' goal is to be ranked as highly as possible in expectation, we argued that our results carry over immediately; and for correlated tests, we explicitly characterized the optimum distribution. However, when the objectives are changed, this ceases to be true. \dkedit{A natural objective is for the principal to maximize the probability of selecing the best product, and for each firm to maximize the probability of being selected. Even for $n=3$ firms, it appears difficult to characterize the equilibria of the endogenous test selection game, or the principal's optimal test distribution.}

\dkedit{Instead of having the principal try to maximize the probability of selecting the better firm, an alternative objective would be for the principal to maximize the expected quality of the selected firm. While this is a natural objective, it requires the model to ascribe meaning to the concrete quality values, rather than using them only for comparison, in contrast to a viewpoint where utilities predominantly encode preferences. Nonetheless, the optimization and equilibrium questions would likely yield a rich set of questions.}

Finally, we note a possibly interesting connection to a very different setting.\footnote{\dkedit{We thank Nicole Immorlica for suggesting this interpretation.}} One can interpret our setting as a principal trying to allocate an item to one of two agents $X,Y$ via a price-discriminating posted-price mechanism. Different from standard such setups, the natural correspondence has a \emph{welfare-maximizing} (rather than \emph{revenue-maximizing}) principal. The mechanism corresponding to our testing setting then has the principal offer the two agents possibly different posted prices. If exactly one agent is interested in buying the item at his posted price, that agent is given the item at the posted price. If both agents are interested in buying at their respective prices, the agent with higher price obtains the item at his posted price. If neither agent is interested in buying, then again, the agent with higher price obtains the item, and pays 0. This model raises the issue of strategic manipulation: an agent might decline the item at his posted price, hoping that his price is higher and he will get the item for free. A natural question is whether the principal can price-discriminate in a way that will provide higher social welfare than offering both agents the same price (and choosing randomly which agent obtains the item if both accept/decline).

\section*{Acknowledgement}
We would like to thank Odilon Camara, Peter Frazier, Moshe Hoffman, Nicole Immorlica, Jonathan Libgober, Erez Yoeli, and Christina Lee Yu for useful discussions and pointers.

SB gratefully acknowledges support from the NSF under grants CNS-1955997, DMS-1839346 and ECCS-1847393.


\bibliographystyle{plainnat}
\bibliography{../davids-bibliography/names,../davids-bibliography/conferences,../davids-bibliography/publications,../davids-bibliography/bibliography,../local-bibliography}





\end{document}